\documentclass[letterpaper, titlepage, 10pt]{article}
\pdfoutput=1
\usepackage[ruled, vlined]{algorithm2e}
\usepackage{setspace}
\usepackage{mathptmx}
\usepackage{graphicx, float}
\usepackage[margin=1in]{geometry}
\usepackage[blocks]{authblk}
\usepackage{subcaption}
\usepackage{amsmath, amssymb, amsfonts}
\usepackage{amsthm}
\usepackage{sectsty}
\usepackage{titlesec}
\usepackage[page]{appendix}
\usepackage{flushend}
\usepackage{stfloats}
\usepackage{xcolor}
\usepackage{mathtools}
\usepackage[urlcolor=blue, pdfborder={0 0 0}]{hyperref}
\newsavebox{\largestimage}

\SetKwInput{KwInput}{Input}
\SetKwInput{KwOutput}{Output}
\SetKwInput{KwInit}{Initialization}
\titlespacing\section{0pt}{12pt plus 4pt minus 2pt}{3pt plus 2pt minus 2pt}
\titlespacing\subsection{0pt}{12pt plus 4pt minus 2pt}{2pt plus 2pt minus 2pt}
\titlespacing\subsubsection{0pt}{12pt plus 4pt minus 2pt}{2pt plus 2pt minus 2pt}

\sectionfont{\fontsize{12}{15}\selectfont}
\newtheorem{theorem}{Theorem}[section]

\newtheorem{proposition}[theorem]{Proposition}
\newtheorem*{remark}{Remark}

\makeatletter
\renewenvironment{abstract}{
      \@beginparpenalty\@lowpenalty
      \small
      \begin{center}%
        \bfseries \abstractname
        \@endparpenalty\@M
      \end{center}}%
     {\par}
\makeatother

\def\email#1{{\tt#1}}
\newcommand{\edited}[1]{\textcolor{black}{#1}}

\begin{document}
\title{Enabling Trade-offs in Privacy and Utility in Genomic Data Beacons and Summary Statistics\thanks{We acknowledge support of this work by the National Institutes of Health (NIH) under grant RM1HG009034 and the National Science Foundation (NSF) CAREER award program under grant IIS-1905558.}}

\author{Rajagopal Venkatesaramani}
\affil{Washington University in St. Louis, St. Louis MO 63130, USA \\ \email{rajagopal@wustl.edu}}

\author{Zhiyu Wan}
\affil{Vanderbilt University Medical Center, Nashville TN 37212, USA \\ \email{zhiyu.wan@vanderbilt.edu}}

\author{Bradley A. Malin}
\affil{Vanderbilt University Medical Center, Nashville TN 37212, USA \\ \email{b.malin@vumc.org}}

\author{Yevgeniy Vorobeychik}
\affil{Washington University in St. Louis, St. Louis MO 63130, USA \\ \email{yvorobeychik@wustl.edu}}

\date{}
\maketitle 

\begin{abstract}
The collection and sharing of genomic data are becoming increasingly commonplace in research, clinical, and direct-to-consumer settings. The computational protocols typically adopted to protect individual privacy include sharing summary statistics, such as allele frequencies, or limiting query responses to the presence/absence of alleles of interest using web-services called Beacons. However, even such limited releases are susceptible to likelihood-ratio-based membership-inference attacks. Several approaches have been proposed to preserve privacy, which either suppress a subset of genomic variants or modify query responses for specific variants (e.g., adding noise, as in differential privacy). However, many of these approaches result in a significant utility loss, either suppressing many variants or adding a substantial amount of noise. In this paper, we introduce optimization-based approaches to explicitly trade off the utility of summary data or Beacon responses and privacy with respect to membership-inference attacks based on likelihood-ratios, combining variant suppression and modification. We consider two attack models. In the first, an attacker applies a likelihood-ratio test to make membership-inference claims. In the second model, an attacker uses a threshold that accounts for the effect of the data release on the separation in scores between individuals in the dataset and those who are not. We further introduce highly scalable approaches for approximately solving the privacy-utility tradeoff problem when information is either in the form of summary statistics or presence/absence queries. Finally, we show that the proposed approaches outperform the state of the art in both utility and privacy through an extensive evaluation with public datasets.

\textbf{Keywords---}Genomic Privacy, Optimization, Anonymity
\end{abstract}

% !TEX root = main.tex

\section{Introduction}
\label{S:intro}

The past several years have seen a sharp rise in the collection and sharing of genomic data as a result of advancements in personalized medicine technology in clinical settings, as well as the rising popularity of direct-to-consumer genetic testing. Data sharing in the former setting is usually controlled through a combination of technical safeguards in order to comply with privacy-protection laws, as well as data-sharing agreements \cite{wan2022sociotechnical}. The latter contributes to sharing of genomic data in both research settings as well as open sharing of data through websites such as OpenSNP \cite{greshake2014opensnp}, where users may, under the guise of anonymity \cite{venkatesaramani2021re}, upload their genome as sequenced by companies such as 23andMe, intending for the data to be useful to researchers and other individuals alike. Commonly used measures to protect individual privacy when sharing genomic data in research settings often involve sharing limited information, for example, queries about the presence or absence of particular single-nucleotide variants \cite{fiume2019federated}, or summary statistics about single-nucleotide polymorphisms~\cite{macarthur2021workshop,sankararaman2009genomic}. Such limited-information releases were initially thought to sufficiently protect individual privacy. However, both presence/absence queries, as well as summary statistics have been shown to be susceptible to membership inference attacks using likelihood ratio tests \cite{homer2008resolving,shringarpure2015privacy}. Typically, it is assumed that an attacker has access to a set of target genomes and leverages statistical tests to infer whether each target individual is present in the dataset. This information about membership in the database can, in turn, be linked to other sensitive information about the individual, based on the metadata associated with the data release.  For instance, a dataset may be known to contain individuals with a certain clinical condition (e.g., heart condition).

A host of techniques have been proposed over the years to protect privacy of released genomic data. Most such methods involve some form of data obfuscation or suppression of a subset of the data, where specific techniques include leveraging the theoretical bounds on the power of inference attacks \cite{sankararaman2009genomic}, federated access control in the form of presence or absence queries \cite{fiume2019federated} or summary statistics \cite{macarthur2021workshop}, optimization-based and game-theoretic approaches \cite{venkatesaramani2021defending,wan2017controlling,wan2017expanding}, and randomization-based techniques, including differential privacy~\cite{dwork2006calibrating}, randomly masking rare alleles \cite{raisaro2017addressing} or simply publishing noisy summary statistics. However, these techniques often do not allow the data custodian to  trade off utility and privacy at sufficiently high resolution, requiring a large utility loss in order to guarantee a desired level of privacy. Further, such methods usually use only one type of data obfuscation---either adding noise to queries or summary statistics, or data suppression---in order to achieve their privacy goals, whereas there may be significant utility gains from combining these, as we show below.
%there may be potential utility gains in combining the two. 

We consider membership-inference attacks on genomic data sharing in two summary release models: 1) summary statistics, where a service publishes, for example, alternate allele frequencies for each genomic variant~\cite{sankararaman2009genomic}, and 2) simple ``minor allele existence'' responses that allow users to query whether a particular allele (e.g., a ``C" or a ``G") is present at a specific position on the genome (e.g., position 1,234,567 on chromosome 8), as done by the Beacon services introduced by the Global Alliance for Genomics and Health (GA4GH) in 2015~\cite{global2016federated}.
%First, we consider Beacon services, introduced by the Global Alliance for Genomics and Health (GA4GH) in 2015 \cite{global2016federated}. Beacon services (or Beacons) are web-services that allow users to query whether a particular allele (e.g., a ``C" or a ``G") is present at a specific position on the genome (e.g., position 1,234,567 on chromosome 8), and the Beacon subsequently responds with \emph{yes} if any individual in the dataset has the queried allele, or with \emph{no} otherwise. Second, we consider the case where genomic summary statistics are published in the form of alternate allele frequencies, i.e. what fraction of individuals in the dataset have an alternate allele at each position on the genome. We further explore two attacker models, following the framework introduced by \cite{venkatesaramani2021defending}: the first involves a fixed-threshold attacker model, where an individual is predicted to be in the dataset if their statistical test score is below a certain constant threshold, calculated \emph{a priori}. The second model is adaptive in that it uses a reference population of individuals who are not in the dataset, and identifies a threshold that best separates these two populations using test scores calculated \emph{after} any defensive measures are implemented. 
We make two contributions.
First, we present a novel model of defense through the lens of an optimization problem that combines query suppression and addition of noise to query responses, and explicitly trades off utility and privacy.
Second, we present two models of attack that inform the privacy component of the utility; the first of these is the conventional approach making use of a fixed threshold to determine membership, while the second makes membership claims adaptively to the defense by choosing a threshold that well separates the individuals in the protected service from those in a reference population.
Third, we present highly scalable algorithmic approaches for both problem settings, and for both threat models, and demonstrate that our approaches 
%We formalize the defense against such membership inference attacks as optimization problems, and propose highly-sclable greedy heuristic approaches to approximately solve them, combining the use of suppression of query responses, and adding noise in the form of flipping Beacon responses or real-valued noise to allele frequencies. We show that for both types of data releases (i.e. Beacon and allele frequencies), as well as for both attack models, our proposed approaches 
improve on the state of the art \emph{both} in utility and privacy, while easily scaling to problem instances involving over 1.3 million genomic variants. 
    % !TEX root = main.tex
\subsection{Preliminaries}
\label{sec:prelims}

\textbf{\textsc{Data Representation}}
\begin{table}[htb]
%\color{red}
\centering
\caption{General Notation}
\begin{tabular}{ll}
\hline
\\[-2ex]
$m$ & The number of SNVs.                                                                           \\
$n$ & The number of individuals.                                                                    \\
$D$ & The dataset of individuals for whom summary information is released.                          \\
$\bar{D}$ & A set of individuals not in the dataset $D$.                                                  \\
$\bar{D}^{(K)}$ & The set of individuals in $\bar{D}$ with LRT scores in the lowest $K$ percentile.             \\
$p_j$ & Alternate allele frequency (AAF) for SNV $j$ for individuals in dataset $D$.             \\
$\bar{p}_j$ & Alternate allele frequency (AAF) for SNV $j$ for individuals in reference set $\bar{D}$. \\
$d_{ij}$ & Binary indicator for whether individual $i$ has an alternate allele for SNV $j$.                  \\
$R^j_n$ & Probability that no individual in $D$ has alternate allele for SNV $j$.                       \\
$x_j$ & Summary release for SNV $j$ - binary for Beacons, real-valued for AAF.                        \\
$\delta$ & Noise added to summary release - binary for Beacons, real-valued for AAF.                     \\
$Q$ & Set of SNVs queried.                                                                          \\
$M$ & Set of SNVs masked.                                                                           \\
$L$ & Likelihood Ratio Test (LRT) score.                                                            \\
$\theta$ & LRT threshold used by the attacker.                                                               \\
$K$ & User-specified percentile of individuals in $\bar{D}$ with the lowest LRT scores.              \\
$Z$ & Set of individuals in $D$ for whom privacy is preserved.                                      \\
$\alpha$ & User specified relative cost of flipping to masking.                                    \\
$w$ & User specified relative weight of privacy versus utility.\\
\\[-2ex]
\hline
\label{tab:notation}
\end{tabular}
\end{table}

A single nucleotide variant (SNV) is a position on the genome where the allele present differs across the population. In this study, all SNVs considered are assumed to be bi-allelic, i.e., there are two possible alleles that may be exhibited for a given SNV. Either one of the possible alleles may be considered a reference allele for a given study, and the other is said to be the alternate allele. The fraction of individuals with the alternate allele in a dataset $D$ of $n$ individuals is referred to as the alternate allele frequency (AAF), which we denote by $p_j$ for the $j^{th}$ SNV. The frequency of the alternate allele in a reference population of individuals not in the dataset (we call this $\bar{D}$) is denoted $\bar{p}_j$. Here, we are only concerned with whether or not an individual has an alternate allele, and not whether both alleles at the chosen position are the alternate allele. Therefore, the binary variable $d_{ij}$ denotes whether individual $i$ has at least one alternate allele ($d_{ij}=1$) at position $j$, and $d_{ij}=0$ otherwise. The total number of SNVs in the dataset is denoted $m$.
Let $Q$ refer to a set of $m$ SNV positions that can be queried. Let $\gamma$ be the genomic sequencing error rate (usually on the order of $10^{-6}$). The probability that no individual in $D$ has an alternate allele (equivalently, all individuals have the reference allele) at position $j$ is given by $R^j_n = (1-\bar{p}_j)^{2n}$. Let the summary release be represented by the vector $x$. In the case of Beacons, $x$ is binary, with $x_j=1$ if $\exists~i\in D, d_{ij}=1$ and $x_j=0$ otherwise - indicating the presence or absence of an alternate allele in the dataset. In the case of summary statistics, the release is a vector of AAFs, therefore, $x_j \in [0,1]$. 

\smallskip
\textbf{\textsc{Membership Inference (MI) Attacks}}\newline
The two membership-inference attacks considered in this work are based on likelihood-ratio test (LRT) statistics. These statistics represent the relative likelihood of an individual $i$ being in the dataset $D$ upon which the summary release was computed, to the likelihood that $i$ is in a reference population. An attacker is assumed to have a set of target genomes, for which membership inference is carried out using released summary information, namely Beacon responses or alternate allele frequencies. In the case of Beacon responses, we use the LRT statistic proposed in \cite{raisaro2017addressing}, which in turn extends the attack originally proposed in \cite{shringarpure2015privacy}. The original attack assumed AAFs to be drawn from the Beta distribution, whereas the extended version uses real AAFs instead. The statistic is computed as follows: let $T$ be the set of target individuals. Then, given the vector $x$ of Beacon responses to queries $Q$, the likelihood-ratio test (LRT) score for individual $i\in T$ is:
\begin{align}
	\label{eq:LRT_x}
	L(Q, d_i, x) = \sum_{j\in Q} d_{ij} &\left(x_j \log \frac{1-R^j_n}{1-\gamma R^j_{n-1}} + (1-x_j)\log\frac{R^j_n}{\gamma R^j_{n-1}} \right)
\end{align}

The LRT statistic for alternate allele frequencies is calculated in a similar fashion and was proposed by \cite{sankararaman2009genomic}. Suppose that we have chosen to release AAFs (real-valued vector $x$) for set $Q$ of SNVs. An attacker who is in possession of the genome of a particular individual $i$ can calculate the log-likelihood ratio statistic for $i$ as follows:
\begin{equation}
	L(Q, d_i, x) = \sum_{j\in Q} d_{ij} \log \frac{\bar{p}_j}{x_j} + (1-d_{ij}) \log \frac{1-\bar{p}_j}{1-x_j}.
	\label{eq:LRT_p}
\end{equation}

\smallskip
\textbf{\textsc{Linkage Disequilibrium}}\newline
So far, the attack models presented are based on the assumption that SNVs are independent. However, in practice, an adversary may be able to exploit correlations between SNVs to infer Beacon responses for SNVs that are modified. The linkage disequilibrium coefficient \cite{von2019re} is one formal measure of such correlations. Given two loci (a locus is the position of a certain gene on a chromosome) with alleles $\{A, a\}$ and $\{B, b\}$ respectively, the linkage disequilibrium coefficient is computed as $LD = P(AB) - P(A) P(B)$.
The value of the linkage disequilibrium coefficient (henceforth called LD) lies between $-0.25$ and $0.25$, where larger values are indicative of higher-than-random association of the specified alleles. 

We evaluate the effect of an adversary accounting for correlations as follows. In this situation, for each SNV that is either flipped or masked, the adversary identifies SNVs that are correlated with the target SNV. Due to the fact that computing the linkage disequilibrium coefficient (LD) for all pairs of 1.3 million SNVs is very computationally expensive, we model an adversary who computes correlations with a fixed number of SNVs on either side of the SNV of interest in the dataset. We consider two SNVs to be correlated if their LD is above a certain threshold, $t_{LD}$, limiting our evaluation to highly correlated SNV pairs. Given a flipped or masked SNV $j$, let $N_{LD}(j)$ be the set of SNVs which are positively correlated with $j$. If at least $75\%$ of SNVs in $N_{LD}(j)$ have a \emph{yes} Beacon response, then the Beacon response for SNV $j$ is inferred to be \emph{yes} as well. We show in a subsequent section that only SNVs with \emph{yes} responses are masked/flipped. Having inferred responses for a subset of the flipped/masked SNVs, the adversary now recomputes LRT scores, and makes membership inference claims using the fixed-threshold or adaptive-threshold models as described above.
Observe that this evaluation is worst-case in the sense that we only perform the above correlation attack for SNVs that \emph{we know} have been flipped or masked, rather than \emph{all} SNVs, as would be done in an actual attack.

\smallskip
\textbf{\textsc{DP and the Laplace Mechanism}}\newline
\edited{Differential privacy (DP) is a popular privacy-preserving data-sharing technique based on the principle that 
the distribution of responses computed on two datasets that only differ in one entry should be similar.
%aggregate information computed on two datasets differing in one entry are sufficiently indistinguishable from each other, such that no information is obtained about any individual entry.
} 
%In other words, the inclusion or deletion of any individual entry should have a negligible impact on the aggregate data. 
Formally, a randomized algorithm $f$ is $\epsilon$-differentially private if for any two datasets $D$ and $D'$ differing in one entry \edited{(i.e., $D'$ omits one record from $D$)},
%\edited{(we assume $|D|=|D'|$, i.e. they contain the same \emph{number} of entries, but differ in the values of one particular entry}), 
and for all possible subsets $F$ of the image of $f$, 
\begin{equation}
	\frac{\textrm{P}(f(D))\in F}{\textrm{P}(f(D'))\in F} \le \epsilon
\end{equation}

\textbf{Unbounded Risk}
In practice, the value of $\epsilon$ is used to trade off utility and privacy, with a smaller value of $\epsilon$ corresponding to a greater degree of privacy, at the cost of utility. For aggregate queries such as means over columns, as is the case with AAFs, a simple mechanism to achieve differential privacy is to add random noise sampled from a Laplace distribution to each SNV's AAF. Laplacian noise with a scale of $\Delta g/\epsilon$ where $\Delta g$ is the sensitivity of the function $g$ (in our case, $g$ is the mean for each SNV) satisfies the requirements for $\epsilon$-differential privacy \cite{dwork2006calibrating}. The sensitivity of a function is defined as  $\max||g(D)-g(D')||_1$, where $D$ and $D'$ differ in one entry \edited{(i.e., $D'$ omits one individual from $D$)}.
%\edited{($|D|=|D'|$)}. 
For a given SNV in a dataset of $n$ individuals, $x_{ij}$ can be either $1$ or $0$.
Therefore, the maximum possible difference between means over columns differing in one entry is $1/n$. As the dataset has $m$ SNVs, the sensitivity, which is the $\ell_1$ norm of the vector of size $n$ with each entry being $1/n$, is $\Delta g = \frac{m}{n}$.
Therefore, adding Laplacian noise centered at $0$, and with a scale of $m/n\epsilon$ satisfies $\epsilon$-differential privacy.

\textbf{Bounded Risk}
While the above measure of sensitivity provides theoretical worst-case privacy guarantees, two genomic sequences rarely - if ever - tend to be completely dissimilar in terms of alternate allele composition. The above measure of sensitivity, owing to the large number of SNVs considered (on the order of 1.3 million), forces the user to choose a very large value of $\epsilon$ (on the order of 100K or above when no SNVs are masked) to retain sufficient practical utility. Such large values of $\epsilon$, in turn, offer lower privacy guarantees. Therefore, we also consider a measure of sensitivity in the average-case scenario, where the numerator consisting of the total number of SNVs, $m$ when calculating sensitivity is replaced by the average number of bits by which a sequence in the dataset differs from each sequence in the reference population. We refer to this as Bounded Sensitivity, and on our data, this measure is an order of magnitude smaller than worst-case sensitivity (on the order of 150K when no SNVs are masked).

    % !TEX root = main.tex
\subsection{Model}
\label{sec:model}

\smallskip
\textbf{\textsc{Defending Against MI Attacks}}\newline
In this work, we make a significant departure from existing approaches to preserving the privacy of individuals when sharing genomic summary statisticcs in two ways. First, we explicitly enable trade-offs between privacy (in the sense of protection from membership inference attacks) and utility (in terms of the extent of modification of the summary statistics). Second, we consider two defensive strategies to mitigate the privacy risks presented by MI attacks - namely suppression (masking) of beacon responses, and the addition of noise to query responses or allele frequencies. In the case of Beacon responses, the addition of noise takes the form of falsifying responses to queries (claiming that a particular alternate allele is not present in the dataset, when in fact, it is). This query-flipping approach is standard in much of the prior literature, with various strategies applied to select the subset of SNVs to flip \cite{wan2017controlling,raisaro2017addressing,cho2020privacy,venkatesaramani2021defending,idash16}. By contrast, masking Beacon responses \cite{sankararaman2009genomic} is a less explored strategy. This is likely because it has a considerably larger impact on utility. In the case of allele frequencies, we add real-valued noise to the published frequencies, which is a hallmark of methods based on differential privacy concepts; but we differ from prior literature in that ours is the first work that combines the addition of noise with suppression of SNVs. Formally, let $M$ be the subset of SNVs that are masked or suppressed. 
\edited{We assume that the data recipients only observe $Q\setminus M$, rather than $Q$ and $M$ separately; consequently, the choice of $M$ does not in itself reveal information about the individuals in the dataset.}
Let $\delta$ denote the noise added to SNVs in $Q\setminus M$. In the case of Beacon responses, let $\delta$ be a vector, with $\delta_j=-1$ indicating that the response for SNV $j$ is flipped, and $\delta=0$ indicating otherwise. To flip a query $j$ is to return a response $1-x_j$ for it, whereas masking $j$ implies that it cannot be queried at all. In case of AAF releases, let $\delta$ denote the real-valued noise added to the AAFs $x$ for SNVs in $Q\setminus M$. 

The LRT score for an individual $i$ when set $M$ of SNVs is masked, and noise $\delta$ is added to the remaining release is $L(Q\setminus M, d_i, x+\delta)$. Overloading notation, we use $L_i(M, \delta)$ to refer to either data release model henceforth - only making a distinction where mathematically necessary. Finally, we can also write the prediction threshold for whether individual $i$ is in the dataset $D$ as a function of the defense, $\theta(M, \delta)$. To ensure privacy is preserved for an individual $i\in D$ is to ensure $L_i(M, \delta) - \theta(M, \delta) \ge 0$. Given $M$ and $\delta$, we define $Z(M,\delta) \subseteq D$ to be the set of individuals for whom masking SNVs in $M$ and adding noise $\delta$ to the rest preserves privacy, i.e.,
\begin{equation}
	Z(M, \delta) = \{i\in D~|~ L(M, \delta) - \theta(M,\delta) \ge 0\}
\end{equation}
Our goal is to solve the following \textsc{Summary Stats Privacy Problem (SSPP)}:
\begin{equation} 
	\min_{M, \delta} \alpha ||\delta||_1 + (1-\alpha)|M| - w|Z(M, \delta)|,
\end{equation}
where $\alpha$ denotes the relative cost of adding noise as compared to masking SNVs, and $w$ captures the relative importance of preserving privacy for individuals over the utility of the released summary statistics. Preserving privacy for all individuals often comes at a prohibitive cost to the utility of the released aggregate data, and is not always desirable. This approach allows the data custodian to explicitly trade off privacy and utility through a combination of masking and adding noise in a systematic way. 

We remark that from a mathematical standpoint, a high-magnitude noise $\delta$ is equivalent to suppression, in the sense that we obtain no useful information about the associated SNV either way. Our distinction here takes the perspective of usability: while these may be equivalent information-theoretically, they are not in the way they are perceived. Specifically, adding a very high level of noise amounts to deception, since one appears to be presenting the information, but in fact it has no value; in contrast, masking is transparent in that no information is actually provided. For this reason, we expect that the relative cost of adding high-magnitude noise is considerably higher than the cost of suppression.

Finally, we consider two threat models, differing in the choice of the prediction threshold, $\theta(M, \delta)$. The first model involves an attacker who computes the threshold \emph{a priori}, not accounting for the defense. We also consider an additional \emph{adaptive} attacker model recently introduced in \cite{venkatesaramani2021defending} in the context of genomic Beacon services. In this stronger adversarial model, the attacker attempts to identify individuals in the dataset, \emph{accounting for the defense} by relying on the separation of LRT scores.

\textbf{Fixed Threshold Attacks} A naive attacker, who does not account for our attempts to defend against such attacks, chooses a threshold $\theta$, typically to balance false positive and false negative rates with respect to some synthetic ground truth dataset. A maximum false positive rate is often used to tune $\theta$ \cite{idash16,wan2017controlling}. This may be accomplished by simulating Beacons on other publicly available datasets, genomic data that the adversary otherwise has access to, or data synthesized using knowledge of alternate allele frequencies. The practical implication of this assumption is that we can set the threshold to be a constant, i.e. $\theta(M, \delta) = \theta$, which does not depend on the subset of SNVs masked or the noise added. As a result, $Z(M, \delta) = \{i\in D~|~L(M, \delta)-\theta \ge 0\}$ for this threat model. This threat model is assumed by most prior work \cite{raisaro2017addressing,sankararaman2009genomic,shringarpure2015privacy}. Note that this model does not require the defender to know the threshold $\theta$ used by the attacker; a conservative bound will do.

\textbf{Adaptive Threshold Attacks} On the contrary, the \emph{adaptive} attack model attempts to separate the two populations (individuals in $D$ from those who are not) using the separation between their LRT scores \emph{after} the defense has been implemented. Recall that $\bar{D}$ is a set of reference individuals not in the dataset $D$. Let $\bar{D}^{(K)}\subset\bar{D}$ be a set of $K$ individuals in $\bar{D}$ such that they have the lowest LRT scores. In the \emph{adaptive} attack model, the prediction threshold is calculated as:
\begin{equation}
	\theta(M, \delta) = \frac{1}{K} \sum_{k\in \bar{D}^{(K)}} L_k(M, \delta)
	\label{eq:theta}
\end{equation}

As a result, the set of individuals for whom privacy is protected, $Z(M, \delta) = \{i\in D~|~L_i(M, \delta)-\frac{1}{K} \sum_{k\in \bar{D}^{(K)}} L_k(M, \delta) \ge 0\}$ under the adaptive attacker model.

% !TEX root = main.tex
\section{Results}
\label{sec:exps}
\textsc{\textbf{Experimental Design}}\newline
\textbf{Dataset and Metrics} Our experiments were conducted on a dataset of 1338843 SNVs on Chromosome 10, made available by the 2016 iDash workshop on Privacy and Security \cite{idash16}. The data consists of genomes of 400 individuals for whom summary statistics are to be released (i.e. the set $D$), and 400 individuals who are not part of this group (the set $\bar{D}$. \edited{This dataset was derived from the 1000 Genomes Project \cite{10002015global} and is sufficiently large to demonstrate the scalability of our approach, while using only Chromosome 10 makes it practical to work with in terms of memory footprint.} All experiments were conducted on a PC with an AMD Threadripper 3960X CPU and 128 GB RAM, running Ubuntu 22.04. We measure utility as $100[1-(\alpha||\delta||_1 + (1-\alpha)|M|)/m]\%$, accounting for the relative cost of adding noise to masking SNVs. We define privacy to be the percentage of individuals for whom privacy is preserved under each respective attacker model. 

\textbf{Baselines - SPG-B} We compare our approach (SPG-B) to four state-of-the-art baselines. The first is strategic flipping (SF), where SNVs are flipped in decreasing order of their differential discriminative power as proposed by \cite{wan2017controlling}, followed by a local search. We also compare to a modified version of SF which we call SFM, where the adaptive threshold definition of privacy is used when applicable. The second is random flipping (RF) \cite{raisaro2017addressing}, where unique alleles in the dataset (i.e., only one individual has the allele) are randomly flipped by sampling from a Binomial distribution. The third is a differentially private (DP) mechanism proposed by \cite{cho2020privacy}, which offers plausible deniability for each Beacon response. Fourth, we compare to the marginal-impact greedy (MIG) approach by \cite{venkatesaramani2021defending}, noting that this approach guarantees privacy for all individuals. For each baseline, we consider a variation that selects SNVs as described by the method, but masks the SNVs instead of flipping them. None of the baselines allow us to combine flipping and masking into a single strategy. %Finally, in a special case of the fixed threshold setting, we are able to compute an optimal solution using an ILP-solver, IBM CPLEX if the search-space is restricted to the SNVs identified by MIG. However, this does not scale when SNVs are masked or in the adaptive setting.

\textbf{Baselines - SPG-R} We compare the performance of the proposed SPG-R algorithm with three baselines: 1) only adding noise using the Laplace mechanism (standard DP), 2) \emph{Masking} only, and 3) the \emph{Linkage} approach proposed by \cite{sankararaman2009genomic}, which greedily selects a subset of SNVs in linkage equilibrium in order of utility up to the maximum allowed power of the LR test. We note that we use the parallel version of SPG-R as detailed in Section~\ref{appx:runtime} of the supplement throughout this work. 

\textsc{\textbf{Fixed Threshold Attacks}}\newline
\textbf{SPG-B} We begin by considering fixed-threshold attacks on Beacons, with $\theta=-250$. For SPG-B, we measure performance in the utility-privacy space by varying the weight $w \in [0.01,10]$. For DP and RF, we vary their respective parameters ($\epsilon$ and probability $p$, respectively). Note that when approaches rely solely on masking SNVs, the only solution for threshold $\theta$ above a certain positive value is to shut the Beacon down (responses which are initially 0 are not masked) - as the maximum LRT score attainable for an individual $i$ is $\theta\ge \sum_{j\in Q_1} d_{ij}B_j$. Therefore, here we show results for a negative value of $\theta$ for which a solution is guaranteed. Results for $\theta=-750$ as well as special cases where all approaches including SPG-B are restricted to either flipping or masking SNVs are similar and provided in Section~\ref{appx:fix_thresh} of the supplement.

\begin{figure}[ht!]
    \centering
        \captionsetup[subfigure]{justification=centering}
        \savebox{\largestimage}{\includegraphics[width=0.45\columnwidth]{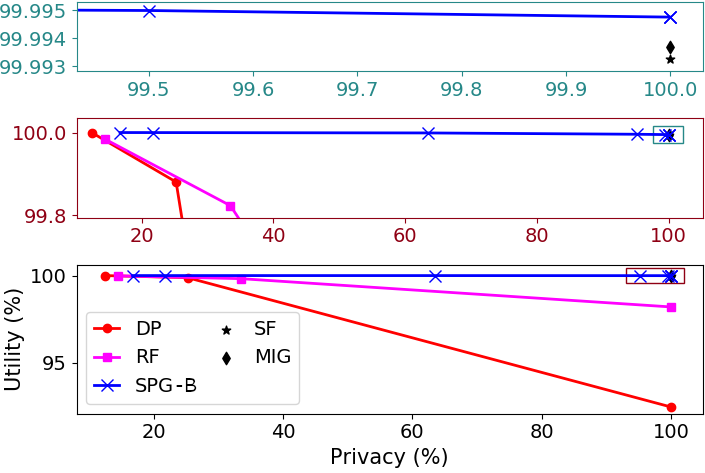}}
        \begin{subfigure}{0.45\columnwidth}
            \usebox{\largestimage}
            \caption{$\theta=-250$, baselines only flip SNVs\\\edited{Zoomed in portions shown in top two subplots.}}
            \label{fig:mdf_flp_fix_n250}
        \end{subfigure}
        \begin{subfigure}{0.45\columnwidth}
            \raisebox{\dimexpr.5\ht\largestimage-.5\height}{
                \includegraphics[width=\textwidth]{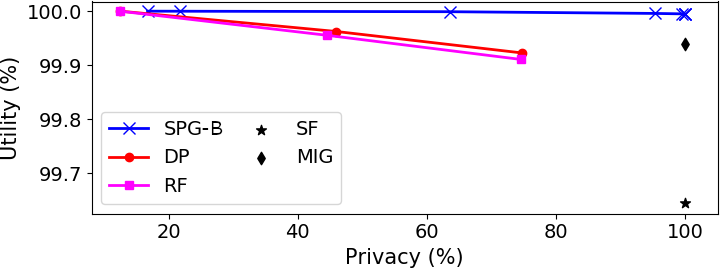}
            }
            \caption{$\theta=-250$, baselines only mask SNVs}
            \label{fig:mdf_msk_fix_n250}
        \end{subfigure}
        \caption{Utility-privacy plots for the fixed threshold attack model for Beacons, compared to baselines.}
        \label{fig:SPG-B_fix}
        \vskip 10pt
\end{figure}

Figs.~\ref{fig:mdf_flp_fix_n250} and ~\ref{fig:mdf_msk_fix_n250} compare the performance of SPG-B to the baseline approaches when a) the baselines only flip SNVs and b) the baselines only mask SNVs, respectively - as none of the baselines allow us to use a combination of the two. The key observation is that the proposed approach Pareto dominates the baselines, as the ability to both flip and mask SNVs provides an additional level of flexibility. The improvement over both DP and RF is particularly substantial in terms of utility. MIG and SF offer slightly lower utility than our approach when the privacy of all individuals is protected, but do not permit solutions that can explicitly trade off utility for privacy.
Fig.~\ref{fig:mdf_msk_fix_n250} tells a similar story, where it can be seen that the difference in utility between SPG-B and MIG increases from about $0.001\%$ (when MIG flipped SNVs) to around $0.05\%$ (in the masking case where $\theta=-250$). This difference in utility corresponds to tens of thousands of more SNVs masked by MIG than SPG-B (as there are over 1.3 million SNVs in the dataset and $\alpha \gg (1-\alpha)$). 

\textbf{SPG-R} In a similar fashion, we compare the parallel variant of SPG-R to the three baselines in the fixed threshold setting. We set $\theta=0$ in these experiments, as this was found to be the threshold that best separated the two populations before the defense was applied. In contrast to Beacons, we explore a wider range of values of $\alpha$, the relative cost of masking to adding noise - as the $\ell_1$ norm of real-valued noise $\delta$ is orders of magnitude smaller on average than it is in Beacons. We vary the weight parameter $w$ in the range $[0.1, 10000]$. For DP, Linkage, and Masking, we choose the privacy-utility point that best optimizes the \textsc{SSPP} objective function for each value of $w$, ensuring fair comparisons. All results involving random noise are averaged over 5 runs. Here, we only present results for unbounded risk; those for bounded risk are similar and provided in Supplement~\ref{appx:bounded}.

We select $\epsilon$ for the Laplace mechanism from $\{10K, 50K, 100K, 500K, 1M, 5M, 10M\}$. While at first, it seems like these values are very large compared to parameters used in practice (on the order of 1-10), this is explained by the fact that the datasets used in practice have nowhere near the number of variables we are dealing with. With over 1.3 million SNVs, using a smaller value of $\epsilon$ would induce a prohibitive amount of noise (refer to Section~\ref{sec:prelims} for details on how noise scales with $m$ and $\epsilon$) that would essentially void the system of any utility.
\begin{figure}[ht!]
    \centering
        \captionsetup[subfigure]{justification=centering}
        \begin{subfigure}[t]{0.45\columnwidth}
            \includegraphics[width=\textwidth]{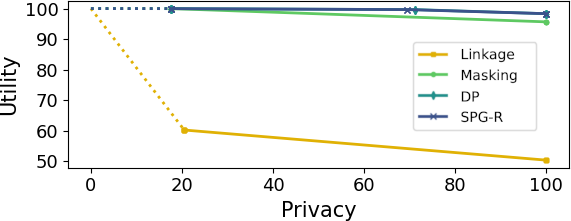}
            \caption{$\theta=0, \alpha=0.5$}
            \label{fig:ssppg_fix_ub_alpha_0.5}
        \end{subfigure}
        \begin{subfigure}[t]{0.45\columnwidth}
            \includegraphics[width=\textwidth]{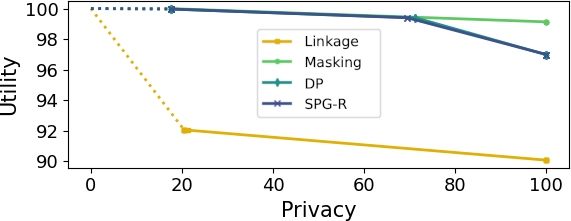}
            \caption{$\theta=0, \alpha=0.9$}
            \label{fig:ssppg_fix_ub_alpha_0.9}
        \end{subfigure}
        \caption{Utility-privacy plots for the fixed threshold attack model for AAF releases, compared to baselines.}
        \label{fig:ssppg_fix}
        \vskip 10pt
\end{figure}

Fig.~\ref{fig:ssppg_fix} considers $\alpha=0.5$ and $\alpha=0.9$. When $\alpha=0.5$, SPG-R outperforms the Linkage and Masking methods, and has comparable performance to DP. When $\alpha=0.9$, on the other hand, Masking  outperforms other methods (since it has far lower cost than adding noise), with both SPG-R and DP having comparable performance (to Masking, and one another).

\textsc{\textbf{Adaptive Threshold Attacks}}\newline
Next, we consider the adaptive threshold attacker, with $K=10$ (here, this refers to the $K$ lowest percentile in terms of LRT scores). Yet again, we present results by varying weight $w$ as before, and the respective parameters for the considered baselines. 

\textbf{SPG-B} Fig.~\ref{fig:SPG-B_ada} presents results in this setting, where we observe that again SPG-B Pareto dominates the baselines, and is comparable to MIG, but now by a much larger margin than in the fixed threshold setting. While SF has better utility than the remaining baselines, it offers very low privacy. If we restrict the baselines to masking only, Fig.~\ref{fig:mdf_msk_ada_10} shows that SPG-B once again outperforms all baselines. The reason is evident from the plot itself - a masking-only strategy is insufficient to guarantee privacy against an adaptive-threshold attacker. Results for $K=5$, as well as settings where all approaches including SPG-B are restricted to flipping or masking, are similar and provided in Section~\ref{appx:ada_thresh} of the supplement.

\begin{figure}[ht!]
    \centering
        \captionsetup[subfigure]{justification=centering}
        \begin{subfigure}[t]{0.45\columnwidth}
            \includegraphics[width=\textwidth]{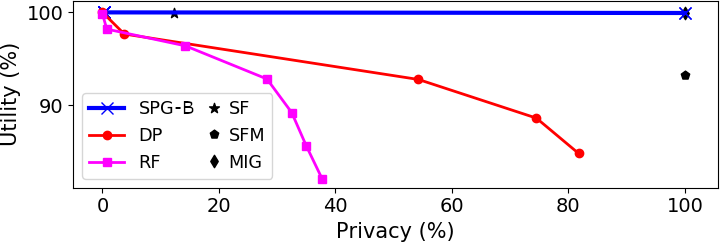}
            \caption{$K=10$, baselines only flip SNVs}
            \label{fig:mdf_flp_ada_10}
        \end{subfigure}
        \begin{subfigure}[t]{0.45\columnwidth}
            \includegraphics[width=\textwidth]{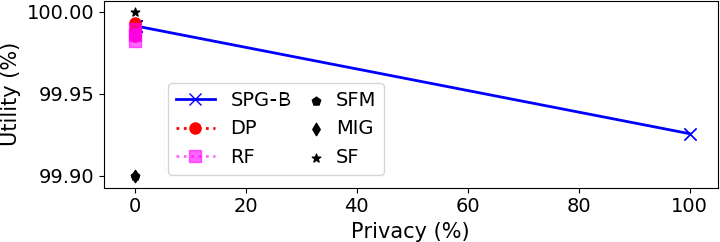}
            \caption{$K=10$, baselines only mask SNVs}
            \label{fig:mdf_msk_ada_10}
        \end{subfigure}
        \caption{Utility-privacy plots for the adaptive threshold attack model for Beacons, compared to baselines.}
        \label{fig:SPG-B_ada}
        \vskip 10pt
\end{figure}

\textbf{SPG-R} Fig.~\ref{fig:ssppg_ada} presents results for $K=10$, for $\alpha=0.75$ and $\alpha=0.9$. In contrast to the fixed threshold setting, here SPG-R dominates all baselines. When $\alpha=0.9$, i.e., adding noise is relatively expensive, Masking produces similar performance when $K=5$ (see supplement, Section~\ref{appx:ada_thresh}); however, once $K$ is increased to $10$, the problem can no longer be solved using Masking alone, and SPG-R dominates it by a significant margin. DP and Linkage offer much lower utility on average compared to only masking SNVs or using SPG-R.

\begin{figure}[ht!]
    \centering
        \captionsetup[subfigure]{justification=centering}
        \begin{subfigure}[t]{0.45\columnwidth}
            \includegraphics[width=\textwidth]{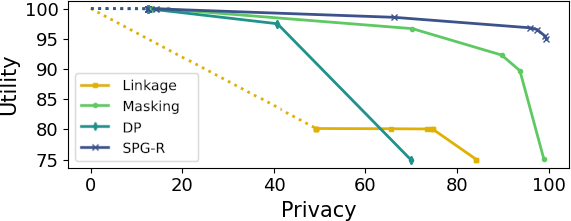}
            \caption{$K=10, \alpha=0.75$}
            \label{fig:ssppg_fix_alpha_0.75}
        \end{subfigure}
        \begin{subfigure}[t]{0.45\columnwidth}
            \includegraphics[width=\textwidth]{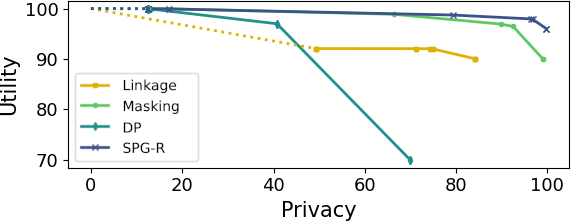}
            \caption{$K=10, \alpha=0.9$}
            \label{fig:ssppg_fix_alpha_0.9}
        \end{subfigure}
        \caption{Utility-privacy plots for the adaptive threshold attack model for AAF releases, compared to baselines.}
        \label{fig:ssppg_ada}
        \vskip 10pt
\end{figure}

\smallskip
\textsc{\textbf{Linkage Disequilibrium}}\newline
Finally, we consider attacks on Beacons which leverage correlations between SNVs. We assume that a pair of SNVs is correlated if their LD-coefficient is above $0.2$. LD is measured within a span of $250$ SNVs on either side of each target SNV. The attack was found to have no impact on DP and RF, so these are omitted from the following plots. Fig.~\ref{fig:LD_fix} presents results in the fixed-threshold setting. Fig.~\ref{fig:LD_fix_flip_1000} shows that the attack has a small impact on the privacy of MIG, and a significant impact on SF and SPG-B - limiting SF to $77-80\%$, and SPG-B to around $75\%$ when baselines flip SNVs for $\theta=1000$. When baselines mask SNVs, the correlation attack has no impact on SF and MIG, but reduces the maximum privacy achieved by SPG-B to around $73\%$ when $\theta=-250$. In both cases, SPG-LD successfully defends against the correlation attack, achieving $100\%$ privacy for large values of $w$, while dominating the baselines. 

\begin{figure}[ht!]
    \centering
    \captionsetup[subfigure]{justification=centering}
    \begin{subfigure}[t]{0.45\columnwidth}
        \includegraphics[width=\textwidth]{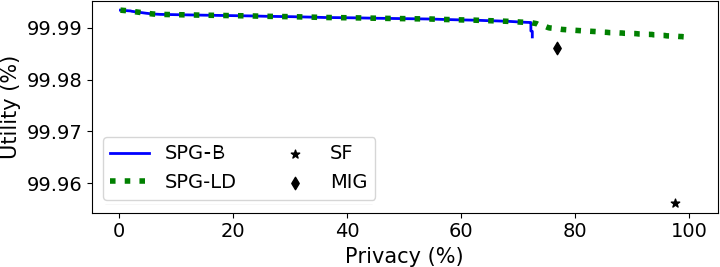}
        \caption{$\theta=1000$, baselines only flip SNVs}
        \label{fig:LD_fix_flip_1000}
    \end{subfigure}
    \begin{subfigure}[t]{0.45\columnwidth}
        \includegraphics[width=\textwidth]{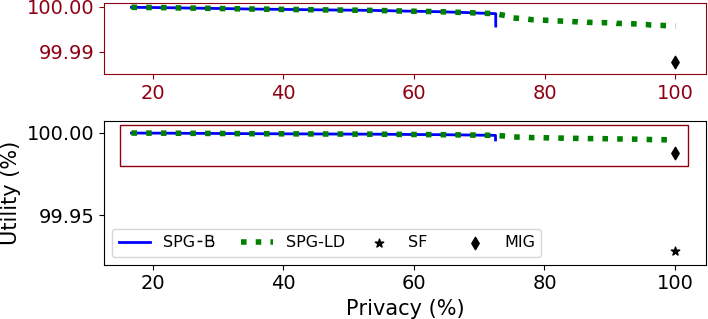}
        \caption{$\theta=-250$, baselines only mask SNVs}
        \label{fig:LD_fix_mask_n250}
    \end{subfigure}
    \caption{Fixed threshold attack model, where the attacker leverages correlation data, and baselines only flip SNVs.}
    \label{fig:LD_fix}
    \vskip 10pt
\end{figure}

Fig.~\ref{fig:LD_ada} presents results in the adaptive threshold case, when $K=5$ and baselines flip SNVs. As before, the attack affects SPG-B, SF and MIG, though to a greater extent in this setting when compared to the fixed threshold model. The privacy achieved by SF drops to about $22\%$ and that of MIG drops to around $80\%$. The privacy achieved by SFM is unaffected by the attack, however SFM yields much lower utility than our proposed methods. The privacy achieved by SPG-B is reduced to about $63\%$ when the attacker uses correlations. The modified approach, SPG-LD, dominates all approaches in terms of utility. In addition, it raises the privacy to about $88\%$. However, it fails to achieve privacy for all individuals - even with very large values of $w$. None of the baselines preserve privacy for any individuals solely by masking SNVs. Therefore, the performance of SPG-LD versus SPG-B is the same as in Fig.~\ref{fig:LD_ada}, such that we do not present new results for the adaptive threshold setting. For additional experiments with relaxed constraints on the choice of $M$, see Supplement~\ref{appx:LD}.
\begin{figure}[ht!]
    \centering
    \captionsetup[subfigure]{justification=centering}
    \begin{subfigure}[t]{0.45\columnwidth}
        \includegraphics[width=\textwidth]{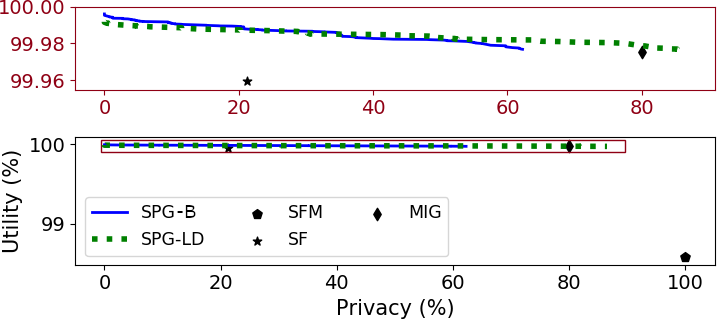}
        \caption{$K=5$, baselines only flip SNVs}
        \label{fig:LD_ada_flip_5}
    \end{subfigure}
    \caption{Adaptive threshold attack model, where the attacker leverages correlation data, baselines flip SNVs.}
    \label{fig:LD_ada}
    \vskip 10pt
\end{figure}
% !TEX root = main.tex
\section{Methods}
\label{sec:methods}

We now present our approach to solving \textsc{SSPP}, under the two threat models discussed above. Note that the noise $\delta$ that is added to summary statistics is qualitatively different in Beacons as compared to AAF summary releases. Recall that while in the former, $\delta$ is additive noise which codifies whether or not SNVs are flipped, in the latter case of AAFs, $\delta$ is real-valued. As a result, the two scenarios yield structurally different optimization problems, but follow the same general framework as outlined in Section~\ref{sec:model}. In both cases, we combine the addition of noise with selective suppression of a subset of SNVs.

We begin by rewriting the LRT scores for individual $i$ as follows. Let $Q_1\subseteq S$ be the subset of SNVs for which Beacon response $x_j=1$, and $Q_0\subseteq S$ be the subset where $x_j=0$. Then, the LRT score for individual $i$ can be written as
\begin{equation}
	\label{LRT_AB_x}
	L(Q, d_i, x) = \sum_{j\in Q_1} d_{ij} A_j + \sum_{j\in Q_0} d_{ij} B_j
\end{equation}

where $A_j = \log \frac{1-R^j_n}{1-\gamma R^j_{n-1}}$ and $B_j = \log\frac{R^j_n}{\gamma R^j_n}$. In case of AAFs, let $A(x_j) = \log \frac{\bar{p}_j}{x_j}, \textrm{ and } B(x_j) = \log \frac{1-\bar{p}_j}{1-x_j}$. Note that in this scenario, $A$ and $B$ are functions of $x_j$, instead of constants for each $j$ as in the beacon service. Then the LRT score can be rewritten as:
\begin{equation}
	\label{LRT_AB_p}
	L(Q, d_i, x) = \sum_j d_{ij} A(x_j) + (1-d_{ij}) B(x_j)
\end{equation}

We note that in the case of Beacons, following \cite{venkatesaramani2021defending}, we assume that the alternate allele is the minor allele at a given position $j$. This, in turn, allows us to leverage a bound on the genomic sequencing error $\gamma$ to ensure our solution approach never violates privacy previously achieved using an iterative process, as we shortly explain. On the other hand, for AAF releases, we only assume that AAFs are bound by $[0.0001, 0.9999]$ in order to prevent division by zero, and any SNV may be masked in case of AAF summary releases. It was shown in \cite{venkatesaramani2021defending} that flipping Beacon responses $x_j$ from 0 to 1 is counterproductive to defending against LRT-based attacks. We now make an analogous observation for masking queries where $x_j = 0$ for both Beacons and AAF summary statistics. 

\begin{proposition}
\label{prop:B_j_ge_0}
In a Beacon service, \edited{for genomic sequencing error} $\gamma<0.25$, $B_j>0$.
\end{proposition}
\begin{proof}
Let $\bar{p}_j$ be the AAF for SNV $j$ in the population. \edited{Recall that $R^j_n = (1-\bar{p}_j)^{2n}$}. As $\bar{p}_j<0.5~\forall j$, $\frac{R^j_n}{R^j_{n-1}} = (1-\bar{p}_j)^2 \ge 0.25$. Since $\gamma<0.25$, $\frac{R^j_n}{\gamma R^j_{n-1}} > 1$, and consequently $B_j = \log\frac{R^j_n}{\gamma R^j_{n-1}} > 0$.
\end{proof}

\begin{proposition}\label{prop:flip_1_to_0}
	Suppose \edited{Beacon response} $x_j=0$ for SNV $j$ given Beacon dataset $D$. Then masking the SNV can never increase the LRT score for an individual $i\in D$, provided $\gamma<\frac{R^j_n}{R^j_{n-1}}~\forall j$.
\end{proposition}
\begin{proof}
	Consider SNV $j$ and an individual $i\in D$. If $d_{ij}=0$ (i.e., the individual does not have a minor allele at position $j$), masking the SNV makes no difference to the LRT score (contribution of $j$ to LRT score is $0$ when $d_{ij}=0$, refer to Eq.~\ref{eq:LRT_x} for details). However, when $d_{ij}=1$, suppressing Beacon response $x_j$ changes the contribution of query $j$ to the LRT score from $\log\frac{R^j_n}{\gamma R^j_{n-1}}$ to $0$. Based on Proposition~\ref{prop:B_j_ge_0}, it can be seen that  $\log\frac{R^j_n}{\gamma R^j_{n-1}} > 0$. Thus, if SNV $j$ is masked the LRT score can only decrease.
\end{proof}

\textsc{\textbf{Masking SNVs}}\newline 
We begin by considering the impact of masking a single SNV on the LRT score. Let $S$ be the set of all SNVs, and $M\subseteq S$ be the subset of SNVs masked. Let $\Delta^M_{ij}$ represent the marginal contribution of masking SNV $j$ on the LRT score for individual $i$. In case of the Beacon service, masking a SNV $j$ changes its LRT score contribution from $d_{ij} A_j$ to $0$, as can be observed from Equation~\ref{eq:LRT_x}. Recall that we only mask SNVs where $x_j=1$, therefore, if the individual does not have the alternate allele (i.e., $d_{ij}=0$), masking the SNV makes no difference. Therefore, for Beacons, $\Delta^M_{ij} = -d_{ij}A_j$. 

Similarly, in case of an AAF summary release, masking an SNV $j$ changes its LRT contribution from $d_{ij} A(x_j) + (1-d_{ij}) B(x_j)$ to $0$, as we can observe from Equation~\ref{eq:LRT_p}. Note that in this case, $\Delta^M_{ij}$ is also a function of the AAF $x_j$, and therefore, as real-valued noise may be added to SNV $j$ as part of our approach before the SNV is masked, $\Delta^M_{ij} = -d_{ij} A(x_j + \delta_j) - (1-d_{ij}) B(x_j + \delta_j)$. On the contrary, we assume that the subsets of SNVs flipped and masked in the case of Beacons are disjoint. 

\smallskip
\textsc{\textbf{Adding Noise to Statistics}}\newline 
Next, we consider the addition of noise to the published statistics, for SNVs that are not masked. Let $\delta$ denote additive noise. In case of the Beacon services, let $\delta_j=-1$ indicate that SNV $j$ is flipped, i.e. the Beacon response for SNV $j$ changes from $1$ to $0$. Note that following the observation in \cite{venkatesaramani2021defending}, we only flip SNVs where initially $x_j=1$. The marginal impact of flipping beacon response for SNV $j$ on the LRT score for individual $i$ is $\Delta^F_{ij} = d_{ij}(B_j-A_j)$, as we can observe from Equation~\ref{eq:LRT_x}. In case of AAFs, we use the Laplace mechanism defined in Section~\ref{sec:prelims} to add real valued noise. Thus in this case, $\delta_j \in [0,1]$, and frequencies after the addition of noise are clipped to ensure they are still in the range $[0.0001, 0.9999]$. In both Beacons and AAFs, the $\ell_1$ norm of $\delta$ quantifies the total amount of noise added to the summary statistics.

\smallskip
\textsc{\textbf{Fixed Threshold Attacks}}\newline
We begin by presenting our solution for the fixed-threshold attack model, where privacy is said to be preserved for an individual $i$ when their LRT score calculated after suppressing set $M$ of SNVs and adding noise $\delta$ lies above a constant prediction threshold $\theta$, specified exogenously.

Let $z_i\in \{0,1\}$ be a binary variable corresponding to individual $i$, where $z_i=1$ when privacy is preserved for $i$, and $z_i=0$ otherwise, and define $y_j=1$ if SNV $j$ is masked (i.e., $j\in M$), and $y_j=0$ otherwise. Then the following optimization problem optimally solves \textsc{SSPP} for the fixed-threshold attacker:
\begin{align}
	\min_{\delta, y, z} \alpha ||\delta||_1 &+ \sum_j (1-\alpha)y_j - w \sum_i\mathclap{z_i} \nonumber\\ 
	 \;&\textrm{subject to:}\nonumber \\ 
	\big(L_i(M, \delta) &- \theta\big)z_i \le 0~\forall~i\in D  \label{eq:ILP_MIP}\\
    y \in\{0,1\}^m, z &\in\{0,1\}^n, \delta \in \begin{cases} \{-1,0\}^m, & \text{Beacons} \\ \mathbb{R}^{|Q\setminus M|}, & \text{AAFs}   \end{cases} \nonumber
\end{align}

In case of Beacons, $\delta$ is an integer vector, with entries being either $-1$ or $0$, and the above optimization problem assumes the form of an integer linear program (ILP). While the ILP optimally solves \textsc{SSPP}, it has an exponential worst-case running time with $\mathcal{O}(3^m)$ possible solutions (each SNV in $Q_1$ can be flipped, masked, or reported truthfully) which poses significant scalability challenges with larger populations over millions of SNVs. In case of AAFs, $\delta$ is real-valued, and thus the above problem becomes a mixed-integer program (MIP). Much like the ILP, the MIP has difficulty scaling to large problem instances with over a million SNVs. To address these limitations, we now introduce heuristic algorithms which approximately solve \textsc{SSPP} for Beacons and AAFs.

\textbf{Heuristic - Beacons}  We now introduce a simple greedy algorithm to compute approximate solutions to \textsc{SSPP} for Beacon services. The driving idea behind our greedy heuristic is as follows: at each iteration, we choose a SNV for which flipping or masking achieves the highest average marginal contribution per unit cost (of flipping or masking the SNV) to the LRT scores for individuals in the Beacon. For each individual, let $P_i$ be the set of SNVs for which the Beacon response $x_j=1$, and the individual's genome has the associated minor allele; i.e., $d_{ij}=1$. For $j\in P_i$, $\Delta^F_{ij}$ and $\Delta^M_{ij}$ are independent of $i$. Let $\Delta^F_j = (B_j-A_j)$ and $\Delta^M_j=-A_j$. Then, $\Delta^F_{ij}=d_{ij}\Delta^F_j$ and $\Delta^M_{ij}=d_{ij}\Delta^M_j$. For a chosen query $j$, and a subset of individuals $P\subseteq B$, let $T_j = \{i\in P | j\in P_i\}$ which is the set of individuals for which $j\in P_i$. The average marginal contribution of flipping the query response to SNV $j$ per unit cost is:
	\begin{equation}
	\label{eq:marginal_flipping}
		\bar{\Delta}^F_j(P) = \frac{|T_j|\Delta^F_j}{\alpha |P|}.
	\end{equation}
    Similarly, the average marginal contribution of masking a SNV $j$ is:
    \begin{equation}
	\label{eq:marginal_masking}
		\bar{\Delta}^M_j(P) = \frac{|T_j|\Delta^M_j}{(1-\alpha)|P|}.
	\end{equation}
	
	At each iteration, we calculate both $\bar{\Delta}^F_j(P)$ and $\bar{\Delta}^M_j(P)$ for every SNV $j$, and either flip or mask the SNV with the highest overall contribution, depending on whether flipping or masking led to it scoring the highest. The number of individuals for whom we thereby guarantee privacy is non-decreasing through each iteration of this algorithm, since flipping or masking SNVs for which $x_j=1$ can only increase LRT scores (see Proposition~\ref{prop:flip_1_to_0}). Each time privacy is assured for at least one additional individual, we compare this privacy-utility point to the current best solution (as measured by the objective function in Equation~\ref{eq:ILP_MIP}), and update it if it improves the objective. We also update $P$ to be the set of individuals for whom privacy is not yet assured. The algorithm iterates until privacy is protected for all individuals in the Beacon, or we cannot flip or mask any more SNVs, at which point we return the overall best solution. This idea is formalized in Algorithm~\ref{algo:SPG}, which we call \textsc{Soft-Privacy-Greedy-Binary} (SPG-B). For runtime analysis, refer Section~\ref{appx:spg_runtime} in the supplement.
 
\begin{algorithm}[ht!]
    \DontPrintSemicolon
    \caption{\textsc{Soft-Privacy-Greedy-Binary (SPG-B)}}
    \KwInput{A set of individuals $i\in D$, subset $P_i$ and LRT score $\eta_i$ for each individual, marginal contributions of flipping/masking $\Delta^F_j$ and $\Delta^M_j$ for each SNV, a set of queries $Q$ and threshold $\theta$, weight parameter $w$, relative cost of flipping $\alpha$.}
    \KwOutput{Subset of queries $F\subseteq S$ to flip, subset of queries $M\subseteq S$ to mask.}
    \KwInit{
	$F=\emptyset$, $M=\emptyset$, $C=\emptyset$, $F_t=\emptyset$, $M_t=\emptyset$, $U=\infty$ \;
    }
	\While{$(D\setminus C)\ne \emptyset$}{
		Set $l=0$, $N=-1$, $d=-1$. \;
		\For{$j\in(Q\setminus (F\cup M))$}{
			Set $T_j=\{i\in (D\setminus C)| j\in P_i\}$ \;
			Set $\bar{\Delta}^F_j = \Delta^F_j \frac{|T_j|}{\alpha|D\setminus C|}$ \;
			Set $\bar{\Delta}^M_j = \Delta^M_j \frac{|T_j|}{(1-\alpha)|D\setminus C|}$ \;
			\If{$\bar{\Delta}^F_j > N$}{
				Set $N=\bar{\Delta}^F_j$ \;
				Set $l=j$ \;
				Set $d=0$ \;
			}
			\If{$\bar{\Delta}^M_j > N$}{
				Set $N=\bar{\Delta}^M_j$ \;
				Set $l=j$ \;
				Set $d=1$ \;
			}
		}
		\If{$d==0$}{
			Set $F_t=F_t \cup l$ \;
		}
		\If{$d==1$}{
			Set $M_t=M_t \cup l$ \;
		}

		\For{$i \in (D\setminus C)$}{
			\If{$\sum_{j\in F}\Delta^F_{ij} + \sum_{j\in M}\Delta^M_{ij} + \eta_i \ge \theta$}{
				Set $C=C\cup i$ \;
			}
		}
		$U_t=\alpha |F_t| + (1-\alpha) |M_t| - w |C|$

		\If{$U_t \le U$}{
			Set $U = U_t$ \;
			Set $F=F_t$ \;
			Set $M=M_t$ \;
		}
	}
	\KwRet $F, M$
	\label{algo:SPG}
\end{algorithm}

\begin{algorithm}[t!]
	\caption{Soft-Privacy-Greedy-Real (SPG-R)}
	\label{algo:SPG-R}
        \DontPrintSemicolon
	
	\KwInput{A set of individuals $i\in D$, marginal contributions of masking $\bar{\Delta}^M_j$ for each SNV, a prediction threshold $\theta$, weight parameter $w$, relative cost of adding noise $\alpha$, number of SNVs to mask per iteration $t$, set $E$ of candidate DP parameters, AAFs $x$ for individuals in $D$ and $\bar{p}$ for individuals in reference set $\bar{D}$.}
	\KwOutput{Subset of SNVs $M\subseteq S$ to mask, real valued noise vector $\delta$.}
	\KwInit{$M=\emptyset$, $C=\emptyset$, $U=\infty$, $\delta=0$, $c=0$, $\Delta^S = $\textsc{Sort}($\bar{\Delta}^M_j$), $M_t=\emptyset$}
        \SetKwFunction{FMain}{GetLR}
        \SetKwProg{Fn}{Function}{:}{}
        \Fn{\FMain{$d_i, \delta, M$}}{
		\KwRet $\sum_{j\in S\setminus M} d_{ij} \log \frac{\bar{p}_j}{x_j+\delta_j} + $ $(1-d_{ij}) \log \frac{1-\bar{p}_j}{1-(x_j+\delta)}$ \;
        }
	\While{$S\setminus M_t \ne \emptyset$}{
		Set $U_t=\infty, \delta_t=0, C_t=\emptyset$\;
		\For{$\epsilon \in E$}{
			Set $\delta_\epsilon = $ Laplacian($0, \frac{|Q\setminus M_t|}{n\epsilon}$) \;
			Set $C_\epsilon=\emptyset$ \;
			\For{$i\in D$}{
				\If{\textsc{GetLR}$(d_i, \delta_\epsilon, M_t)\le \theta$}{
					Set $C_\epsilon = C_\epsilon \cup i$ \;
				}
			}
			Set $U_\epsilon = \alpha||\delta_\epsilon||_1 + (1-\alpha)|M_t| - w |C_\epsilon|$ \;
			\If{$U_\epsilon \le U_t$}{
				Set $U_t = U_\epsilon$ \;
				Set $\delta_t = \delta_\epsilon$ \;
				Set $C_t = C_\epsilon$ \;
			}
		}
		\If{$U_t \le U$}{
			Set $U=U_t$ \;
			Set $\delta=\delta_t$ \;
			Set $M=M_t$ \;
		}
		Set $ct=1$ \;
		\While{$ct\le t$}{
			Set $M_t=M_t\cup \Delta^S_c$ \;
			Set $c=c+1$ \;
			Set $ct=ct+1$ \;
		}
		Set $U_t = \alpha||\delta_t||_1 + (1-\alpha)|M_t| - w |C_t|$ \;
		\If{$U_t \le U$}{
			Set $U=U_t$ \;
			Set $\delta=\delta_t$ \;
			Set $M=M_t$ \;
		}
	}
	\KwRet $M, \delta$ \;
\end{algorithm}

\textbf{Heuristic - AAFs} We now introduce an alternating optimization algorithm which approximately solves \textsc{SSPP}, combining masking of alternate allele frequencies for a subset of SNVs with adding Laplacian noise to the rest. 

The outline of the algorithm is as follows. 
At each step, we alternate between adding noise to SNVs that have not yet been masked, such that it minimizes the objective in Equation~\eqref{eq:ILP_MIP}, and masking SNVs in order of their average marginal contribution to the population's LR scores, in an attempt to increase utility. The average marginal contribution of masking a SNV $j$ is simply the mean over the marginal contributions for all individuals, i.e.,
%\begin{equation}
\(	\bar{\Delta}^M_j = \frac{1}{n} \sum_{i\in D} \Delta^M_{ij}. \)
%\end{equation}

Since in the fixed-threshold model, we require that each individual's score lie \emph{above} a specified threshold, we aim to increase LRT scores. Thus, we rank SNVs to mask in order of their average marginal contributions. For $\alpha \gg (1-\alpha)$, masking is preferred over adding noise---in essence, sharing a smaller subset of cleaner data, as opposed to sharing all SNVs with high obfuscation. Masking SNVs should necessarily continue to minimize the objective function until privacy is violated for an individual previously covered as a result of encountering large positive values of $\Delta_j$. However, continuing to mask in this manner is suboptimal because it does not allow us to explore possible intermediate solutions; for example, masking fewer SNVs and adding slightly higher noise may provide a better privacy-utility trade-off in many cases. 

Our algorithm proceeds as follows. We first add noise to all SNVs by sampling from the Laplace distribution such that it best optimizes the objective in Equation~\eqref{eq:ILP_MIP}. This may be done in one of two ways: 1) computing the objective over a pre-selected set of values of $\epsilon$, or 2) a binary search over possible values of $\epsilon$, assuming convergence when the difference between two considered values of $\epsilon$ in the search is sufficiently small. While the latter is more systematic, it is also slower and performs poorly (refer Section~\ref{appx:runtime}), while the former approach is fully parallelizable and produces good results, as long as the choices for the set of candidate $\epsilon$ values are reasonable. We therefore use the first approach in the rest of this work.

Having added Laplacian noise, we then mask a set of $t$ SNVs in the order of their average marginal contribution to LR scores, calculated after adding noise. The value of $t$ that we use is chosen to balance computation time and the near-optimality of the solution. Specifically, a smaller value of $t$ implies a larger number of candidate solutions explored, but with a runtime inversely proportional to $t$. 

At the end of this cycle, we repeat the noise-addition and masking processes in an alternating fashion, each time adding noise to the SNVs that remain unmasked with the scale of the Laplacian distribution accordingly adjusted. Algorithm~\ref{algo:SPG-R} - which we call the \textsc{Soft-Privacy-Greedy-Real} (SPG-R) approach - provides full details about the implementation of our method. In Section~\ref{appx:runtime} in the supplement, we present an alternate implementation of the SPG-R algorithm which leverages problem structure to reduce redundant computations, and parallel processing in order to significantly reduce runtime.

\smallskip
\textsc{\textbf{Adaptive Threshold Attacks}}\newline
In the \emph{adaptive} threshold scenario, the goal is to ensure that the LRT scores of individuals in $D$ and $\bar{D}$ (i.e., those not in the dataset) remain sufficiently well-mixed. Recall from Section~\ref{sec:prelims} that the prediction threshold in this setting is $\theta(M, \delta) = \frac{1}{K} \sum_{k\in \bar{D}^{(K)}} L_i(M, \delta)$, which is the average LRT score for a set of $K$ individuals in $\bar{D}$ with the lowest LRT scores. Then similar to Eq.~\ref{eq:ILP_MIP}, we can formulate this as an optimization problem in the context of adaptive attacks.

\begin{align}
    \centering
	\min_{\delta, y\in\{0,1\}^m, z\in\{0,1\}^n} \alpha ||\delta||_1& + \sum_j (1-\alpha)y_j - w \sum_i\mathclap{z_i} \nonumber\\ 
	 \;&\textrm{subject to:}\nonumber \\ 
	\big(L_i(M, \delta) - \theta(M, &\delta) \big)z_i \le 0~\forall~i\in D  \label{eq:MIP}\\
        \delta \in \{-1,0\}^m \text{(Beacons)}&;\quad \delta \in \mathbb{R}^m \text{(AAFs)}    \nonumber
\end{align}

This structure allows us to extend our algorithms used for the fixed-threshold scenario, with one change - instead of sorting SNVs by $\Delta^M_j$ or $\Delta^F_j$, we now sort the SNVs by $\Delta^{M(K)}_j = \Delta^M_j - \frac{1}{K} \sum_{k\in \bar{D}^{(K)}} \Delta^M_{kj}$ and $\Delta^{F(K)}_j = \Delta^F_j - \frac{1}{K} \sum_{k\in \bar{D}^{(K)}} \Delta^F_{kj}$ respectively. In the adaptive threshold model, with Beacons, $\Delta^{M(K)}_j$ and $\Delta^{F(K)}_j$ may be negative, and may be detrimental to privacy achieved in prior iterations of our greedy algorithms. As such, masking and flipping are respectively restricted to those SNVs where these quantities are strictly positive. 

\textsc{\textbf{Linkage Disequilibrium}}\newline
To defend against an attacker who leverages correlations to infer flipped/masked SNVs in a Beacon, we introduce a direct extension to our proposed \textsc{Soft-Privacy-Greedy} approach. Specifically, whenever an SNV $j$ with known correlations is flipped or masked, all SNVs correlated to it (the set $N_{LD}(j)$) are also flipped or masked, respectively. To capture the corresponding utility loss while deciding which SNV to flip or mask, we modify the \textsc{Soft-Privacy-Greedy} algorithm as follows. For each SNV $j$, we amend the marginal contribution of flipping $j$ to $\bar{\Delta}^F_j(P) = \frac{T_j\Delta^F_j}{\alpha |P| |N_{LD}(j)|}$, and the marginal contribution of masking SNV $j$ to be $\bar{\Delta}^M_j(P) = \frac{T_j\Delta^M_j}{(1-\alpha) |P| |N_{LD}(j)|}$.
The algorithm then proceeds as before, with the added condition that, any time a SNV $j$ is picked such that $N_{LD}(j) \ne \emptyset$, all SNVs in $N_{LD}(j)$ are also correspondingly flipped or masked. We refer to this modified algorithm as \textsc{Soft-Privacy-Greedy-LD} or SPG-LD.

% !TEX root = main.tex
\section{Discussion}

In this study, we presented a formalization to the problem of finding the optimal privacy-utility tradeoff when defending against membership-inference attacks on genomic summary releases (Beacon services and summary statistics), allowing - unlike prior studies - for the defense to combine masking of SNVs and the addition of noise to best balance the two, while accounting for the relative cost of adding noise as compared to suppressing responses. In the case of Beacons, We further evaluate an extension of the proposed approach against a more powerful attacker model where correlations between SNVs are exploited to infer modified responses. We present a simple yet principled greedy algorithm for both release models to discover the best privacy-utility balance which outperforms prior art, evaluating it against powerful attacks from recent literature. It should be recognized that our approach does have certain limitations, in that it is specific to the MI attacks that leverage an LRT score. More powerful attacks may be devised that defeat our approach. 
\edited{\section{Data Access}
All of the code and data used in this study are publicly available at the following public repository:  \url{https://doi.org/10.5281/zenodo.7510802}. }
\edited{\section{Competing Interests Statement}
The authors declare that they have no competing interests with the research communicated in this paper.}
\edited{\section{Acknowledgements}
The authors would like to thank Ashwin Kumar at Washington University in St. Louis for valuable inputs and discussions.}

\bibliographystyle{splncs04}
\bibliography{biblio}

\clearpage
\appendix
\renewcommand{\appendixpagename}{Supplementary Material}
\appendixpage

% !TEX root = main.tex

This supplement provides a comprehensive overview of related work, additional results for the various threat models described in the paper, runtime analysis for the SPG-B algorithm. and a parallel implementation of SPG-R which leverages problem structure to speed up computation.

\section{Related Work}
\label{relwork}
The prior literature on data privacy spans a host of techniques. There has been a substantial amount of research into mechanisms based on differential privacy \cite{dwork2006calibrating}. Recently, for instance, local differential privacy (LDP) \cite{erlingsson2014rappor} has gained popularity, where data is first perturbed locally before an aggregator computes overall database statistics. Wang and colleagues \cite{8731512} proposed novel LDP mechanisms, and extended them to aggregate computations on locally perturbed multidimensional data. Gu et al. \cite{gu2020providing} addressed the potential differences in privacy requirements at each data collection source for LDP with an input-discriminative extension to LDP. Takagi et al. \cite{9458927} proposed P3GM, a differentially-private generative model based on variational autoencoders to overcome the issue of a large amount of noise injected into high-dimensional data by traditional DP techniques. Xie and colleagues \cite{xie2021generalized} extended an encryption scheme used for IP addresses to more general datatypes, combined with a multi-view outsourcing method which generates one utility-preserving view of the data for analysis among several fake indistinguishable views in order to protect privacy. 

Attacks specific to Beacon services have evolved since their introduction in 2015. Shringarpure and Bustamante \cite{shringarpure2015privacy} illustrated how to leverage likelihood-ratio test (LRT) scores to make membership inference claims from Beacon responses. They specifically used Beacons from the 1000 Genomes Project \cite{10002015global} and the Personal Genome Project \cite{chervova2019personal} and demonstrated that only a small number of queries are sufficient to predict membership. The study assumed that allele frequencies are drawn from a Beta distribution. Raisaro and colleagues \cite{raisaro2017addressing} extended the attack using real allele frequencies instead of assuming that allele frequencies are drawn from the Beta distribution. The authors proposed three defense strategies against such a membership inference attack: a) a Beacon alteration strategy, where the Beacon flips all responses for unique alleles (i.e., only one individual in the dataset contains a minor allele at a given position), b) a random flipping strategy, where the Beacon responses for unique alleles are flipped, but randomly by sampling from a binomial distribution, and c) a query-budget strategy, where the contribution of each individual in the dataset to the Beacon responses is used to decide whether the individual's genome will be included in providing a Beacon response to authenticated Beacon users. In the query-budget approach, however, the authors assumed that the genomic sequencing error is 0, in order to simplify analysis. Greedy approaches in this special case with no sequencing error are discussed in \cite{venkatesaramani2021defending}.

Von Thenen and colleagues \cite{von2019re} introduced an allele-inference technique by leveraging linkage disequilibrium between alleles, where a higher-order Markov chain is used to infer alleles at positions of interest from a few correlated SNVs. The study showed that, by inferring hidden SNVs using the proposed technique, an attacker can make membership-inference claims, despite making far fewer queries to the Beacon. This allows the attacker to potentially bypass a query-budget defense, as well as defenses that mask SNVs with smaller minor allele frequencies. A related attack, introduced by Ayoz and colleagues \cite{ayoz2021genome}, considers genome reconstruction for evolving Beacons, where Beacon-responses are supplemented with phenotype metadata, and the attacker already knows about an individual's membership. Bu and colleagues  \cite{bu2021haplotype} introduced a haplotype-based membership inference attack that reconstructs haplotypes using allele frequencies, as opposed to relying on a target genome. Samani and colleagues \cite{samani2015quantifying} introduced am method that relied on high-order SNV correlations to carry out an inference attack on  Beacons using Markov models.

At the same time, several defenses against such attacks have been proposed. Wan and colleagues \cite{wan2017controlling} selected a subset of SNVs to flip by defining a differential discriminative power that captures a SNV's marginal contribution to the LRT score. The approach selected SNVs in decreasing order of the proposed metric, followed by a greedy local search to improve utility. Cho and colleagues \cite{cho2020privacy} used an approach based on differential privacy to flip Beacon responses, a method that forms one of the baselines we use in this paper. This study relied upon a differentially private geometric mechanism, treating SNVs in an independent manner. Finally, \cite{venkatesaramani2021defending} proposed greedy algorithms to select a subset of SNPs based on marginal impact, by drawing parallels to the set-cover problem. This study identified several special cases, such as very low DNA sequencing error, and when allele frequencies were assumed to be drawn from the Beta distribution. 

\section{SPG-B Runtime Analysis}
\label{appx:spg_runtime}
We show that the running time of the proposed \textsc{SPG-B} algorithm is quadratic in the number of SNV queries $m$ and linear in the number of individuals $m$ in the Beacon.

\begin{theorem}
The worst-case running time of the \textsc{Soft-Privacy-Greedy-Binary} algorithm is $\mathcal{O}(m^2 n)$.
\end{theorem}
\begin{proof}
    All operations inside the first \emph{for} loop are constant time, and the loop executes $m$ times, where $m$ is the number of SNVs - therefore yielding a complexity of $\mathcal{O}(m)$. The \emph{if} condition inside the second \emph{for} loop involves a sum over $m$ SNVs, and is therefore an $\mathcal{O}(m)$ operation, and this \emph{for} loop executes at most $n$ times in the worst case, with a total time complexity of $\mathcal{O}(mn)$. The remaining statements within the \emph{while} loop are constant time operations. The outer \emph{while} loop executes at most $m$ times, as all SNVs are flipped or masked in the worst case, and therefore the overall time complexity of SPG-B is $\mathcal{O}(m (m + mn + 1)) = \mathcal{O}(m^2 n)$. 
\end{proof}

\begin{remark}
In practice, $m\gg n$, therefore we can treat $n$ as a small constant, and the runtime is approximately $\mathcal{O}(m^2)$. \edited{We further note that the worst-case running time complexity of MIG and SPG-B is of the same magnitude, as they proceed almost identically, except that in SPG-B we consider marginal impacts of masking, as well as flipping each SNV, and compare the two - which are constant time operations inside a loop that iterates over all SNVs.}
\end{remark}

\section{Improving Runtimes for SPG-R}
\label{appx:runtime}
When optimizing the real valued noise using the Laplace mechanism, we can utilize some structural observations to avoid redundant computations over large matrices which, in turn, can cut down runtimes by orders of magnitude. Our first observation is that the scale of the Laplacian distribution from which random noise is drawn is directly proportional to the number of SNVs that are not yet masked. Recall that the scale of the Laplacian when no SNVs are masked is $m/n\epsilon$. If a $k^{th}$ fraction of SNVs remains unmasked, the scale of the Laplacian, accordingly, is $m/kn\epsilon$. Consider the noise added per SNV for a given value of $\epsilon$, when no SNVs are masked. This same amount of noise is achieved \emph{per SNV} when a $k^{th}$ fraction of SNVs is unmasked, with noise added corresponding to $\epsilon/k$, and herein lies our first runtime improvement. When a noise sample is drawn from a Laplacian distribution with scale $s$ for a given value of $\epsilon$, the same noise can be used at scale $s/k$ for a corresponding DP parameter $\epsilon/k$. 
\begin{algorithm}[ht!]
	\caption{The SPG-R (parallel) Algorithm}
        \DontPrintSemicolon
        \KwInput{A set of individuals $i\in D$, a prediction threshold $\theta$, weight parameter $w$, marginal contributions of masking $\Delta^M$, relative cost of adding noise $\alpha$, number of SNVs to mask per iteration $t$, set $E$ of candidate DP parameters, AAFs $x$ for individuals in $D$ and $\bar{p}$ for individuals in reference set $\bar{D}$.}
        \KwOutput{Subset of SNVs $M\subseteq Q$ to mask, real-valued noise vector $\delta$.}
        \KwInit{$M=\emptyset$, $C=\emptyset$, $U=\infty$, $\delta=0$}
        \SetKwFunction{FMain}{GetLR}
        \SetKwProg{Fn}{Function}{:}{}
        \Fn{\FMain{$x_i, \Delta, M$}}{
		\KwRet $\sum_{j\in Q\setminus M} -\Delta_{ij}$
        }
        \For{$\epsilon \in E$, \textbf{\upshape in parallel}}{
        Set $\delta^\epsilon = $ Laplacian($0, \frac{|Q|}{n\epsilon}$), $M_\epsilon = \emptyset$, $c_\epsilon=0$\\
            $\Delta^\epsilon_{ij} = -d_{ij}\log\frac{\bar{p}_j}{x_j + \delta^\epsilon_j} - (1-d_{ij}) \log\frac{1-\bar{p}_j}{1-(x_j + \delta^\epsilon_j)}$\\
            Set $\Delta^\epsilon_j = \frac{1}{|D|}\sum_{i\in D} \Delta^\epsilon_{ij}$\\
            Set $\Delta^{S\epsilon} = $\textsc{Sort}($\Delta^\epsilon_j$)\\
            \While{$Q\setminus M_\epsilon\ne\emptyset$}{
                Set $C_\epsilon=\emptyset$\\
                \For{$i\in D$}{
                    \If{\textsc{GetLR}($x_i, \Delta^\epsilon, M_\epsilon$)$\le \theta$}{
                        Set $C_\epsilon = C_\epsilon \cup i$\\
                    }
                }
                Set $U_\epsilon = \alpha||\delta_\epsilon||_1 + (1-\alpha)|M_\epsilon| - w|C_\epsilon|$\\
                \If{$U_\epsilon \le U$}{
                    \textsc{AcquireLock($U, \delta, M$)}\\
                    Set $U=U_\epsilon$\\
                    Set $\delta=\delta^\epsilon$\\
                    Set $M=M_\epsilon$\\
                    \textsc{ReleaseLock($U, \delta, M$)}\\
                }
                Set $ct=1$\\
                \While{$ct\le t$}{
                    Set $M_\epsilon = M_\epsilon \cup \Delta^{S\epsilon}_c$\\
                    Set $c=c+1$\\
                    Set $ct=ct+1$\\
                }
            }
    }
    \Return $M, \delta$\\
    \label{algo:SPG-R-parallel}
\end{algorithm}
While drawing a noise sample is not an expensive operation in itself, this observation allows us to re-use previously computed values for $\Delta^M_{ij}$ and $\Delta^{M(K)}_{ij}$ at different (scaled) values of $\epsilon$ as more SNVs are masked. As both $\Delta^M$ and $\Delta^{M(K)}$ contain one entry per individual per SNV, re-computing values for a fresh noise sample each time contributes significantly to overall runtime. As long as the choices of $\epsilon$ at the beginning of the algorithm are well-spread out (over multiple orders of magnitude, as the best solutions may involve masking a significant fraction of SNVs), we can avoid re-computing  $\Delta^M$ and $\Delta^{M(K)}$ as we change the scale of the Laplacian, instead assuming the noise to be generated for a correspondingly scaled value of $\epsilon$. We note that the $\ell_1$ norm of the noise for the objective function would still have to be recomputed over only the SNVs that remain unmasked, but this is a relatively inexpensive operation. 

The second structural observation about our solution approach is that the privacy-utility points explored for a given set of candidate $\epsilon$ values are independent of $w$, the relative importance of guaranteeing privacy over preserving utility. In Algorithm~\ref{algo:SPG-R}, threads for parallel processing are initialized after masking every $t$ SNVs. Even if we use global variables (one instance of noise $\delta$ and  $\Delta^M_{ij}$ or $\Delta^{M(K)}_{ij}$, depending on the attacker model, for each value of $\epsilon$), the repeated creation and synchronization of threads before masking the next set of SNVs can add significant overhead. To deal with this, we re-formulate our search on a per-$\epsilon$ basis, where each thread masks SNVs locally. Threads still share access to the matrix $x$, but avoid repeated function calls and synchronization wait times. Moreover, we can save all candidate solutions explored by recording $||\delta||, \sum_j y_j$, and $\sum_i z_i$ under either attack model, and for any value of $\alpha$ and $w$, and find the best solution in linear time over the saved candidate solutions, which are in turn linear in the number of SNVs. The algorithm, SPG-R (parallel), that takes advantage of these improvements, is provided in Algorithm~\ref{algo:SPG-R-parallel}.

\smallskip
\textsc{\textbf{Optimizing Over $\epsilon$}}\newline
We compare the two variants of SPG-R (binary and parallel) under the harder of the two attacker models - the adaptive threshold setting. 
	\begin{figure}[h!]
        \centering
        \captionsetup[subfigure]{justification=centering}
        \begin{subfigure}[t]{0.45\columnwidth}
            \includegraphics[width=\textwidth]{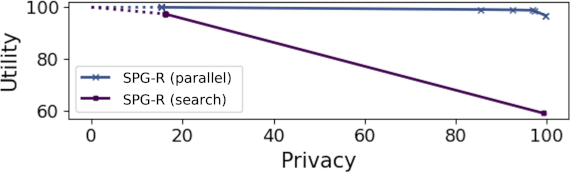}
            \caption{K=5}
            \label{fig:bin_V_par_5}
        \end{subfigure}
        \begin{subfigure}[t]{0.45\columnwidth}
            \includegraphics[width=\textwidth]{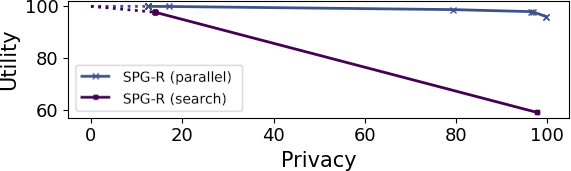}
            \caption{K=10}
            \label{fig:bin_V_par_10}
        \end{subfigure}
        \caption{Relative performance of SPG-R (binary) and SPG-R (parallel) under the adaptive threshold model}
        \label{fig:bin_v_parallel}
        \vskip 10pt
    \end{figure}
Fig.~\ref{fig:bin_v_parallel} shows the relative performance of the two variants when $K=5$ and $K=10$. Binary search was initialized with $\epsilon \in $\{10K, 10M\}, and candidate values for SPG-R (parallel) were selected from $\epsilon \in $\{10K, 50K, 100K, 500K, 1M, 5M, 10M\}. 
We can observe that SPG-R (parallel) significantly outperforms SPG-R (binary) in utility-privacy tradeoff, likely because the latter does not explore useful tradeoff points.

\smallskip
\textsc{\textbf{Empirical Runtime Comparison}}
    \begin{figure}[h!]
    \centering
    \includegraphics[width=0.55\columnwidth]{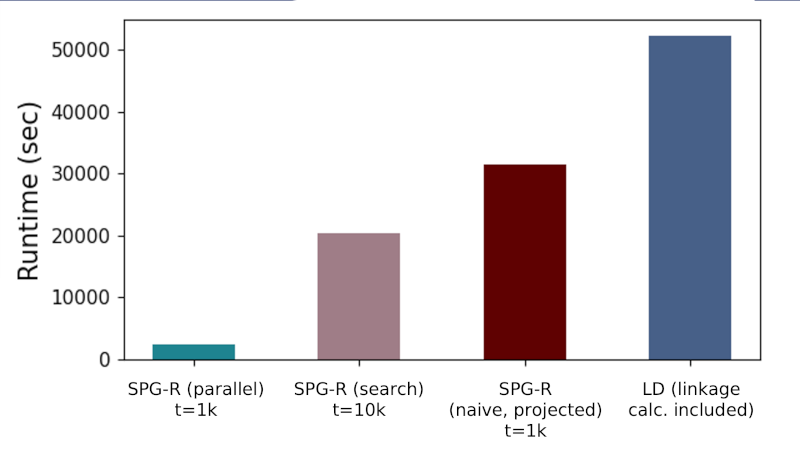}
    \caption{Runtime comparison between SPG-R variants and baselines.}
    \label{fig:spg_runtimes}
    \vskip 10pt
\end{figure}
Next, we compare the runtimes for the various methods used in Fig.~\ref{fig:spg_runtimes}. 
The number of SNVs masked in one iteration ($t$) used for each approach is indicated in the plot. We also compare the runtimes to a projected estimate of a non-parallel naive implementation of SPG-R, where solutions over the various candidates for $\epsilon$ are sequentially computed. Runtime for DP is omitted because it is too small compared to the rest. The runtime for SPG-R (binary) is an order of magnitude larger than SPG-R (parallel), even when masking 10 times the number of SNVs in each iteration ($t=10K$), taking about 5.5 hours in practice. SPG-R (binary) with $t=1000$ can therefore be expected to take in excess of 55 hours. The estimated runtime for naive implementation of SPG-R is calculated by multiplying the average runtime of SPG-R (parallel) with the number of threads and adding some marginal overhead for thread creation and synchronization. Runtime for Linkage includes the time taken to compute linkage disequilibrium coefficients for pairs of SNVs using a sliding window of 500 SNVs (250 on either side of each SNV), which takes about 15 hours in practice, although we note that this is a parallelizable problem with scope for shared data structures, and the computation only needs to be done once. 

\section{Additional Results - Fixed Threshold}
\label{appx:fix_thresh}
First, we present additional results where SPG-B combines flipping and masking, while baselines solely flip or mask SNVs for $\theta=-750$ in Fig~\ref{fig:SPG-B_fix_n750}. Performance, in this case, is similar to the results for $\theta=-250$, in that SPG-B Pareto dominates all baselines, and shows significant improvements in utility over DP and RF - both when baselines only flip SNVs (Fig.~\ref{fig:mdf_flp_fix_n750}) or mask SNVs (Fig.~\ref{fig:mdf_msk_fix_n750}).

\begin{figure}[ht!]
    \centering
        \captionsetup[subfigure]{justification=centering}
        \savebox{\largestimage}{\includegraphics[width=0.45\columnwidth]{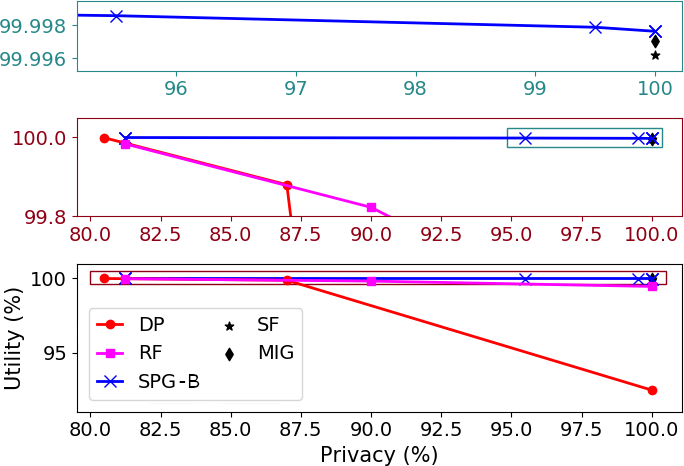}}
        \begin{subfigure}[t]{0.45\columnwidth}
            \usebox{\largestimage}
            \caption{$\theta=-750$, baselines only flip SNVs}
            \label{fig:mdf_flp_fix_n750}
        \end{subfigure}
        \begin{subfigure}[t]{0.45\columnwidth}
            \raisebox{\dimexpr.5\ht\largestimage-.5\height}{
                \includegraphics[width=\textwidth]{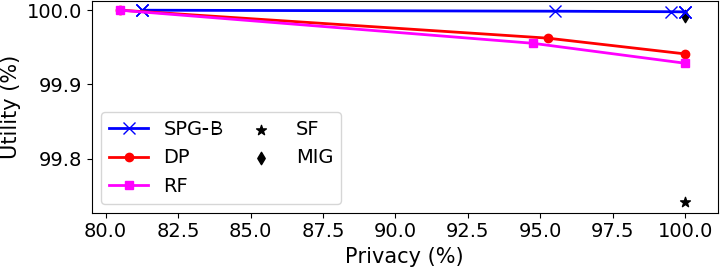}
            }
            \caption{$\theta=-750$, baselines only mask SNVs}
            \label{fig:mdf_msk_fix_n750}
        \end{subfigure}
        \caption{Utility-privacy plots for the fixed threshold attack model for Beacons, compared to baselines.}
        \label{fig:SPG-B_fix_n750}
        \vskip 10pt
\end{figure}

\smallskip
\textsc{\textbf{Special Case: Only Flipping SNVs}}\newline
We now consider the special case where all approaches, including SPG-B, are restricted to flipping SNVs.
This represents a scenario where suppressing Beacon responses may be impractical. In this setting, in addition to SPG-B and the various baseline methods shown in the more general setting, we also present the optimal solution computed using CPLEX \cite{cplex2009v12}, an ILP-solving toolkit. Since the original ILP in~\eqref{eq:MIP} is unable to scale to a search space consisting of 1.3 million SNVs, we restrict the ILP to search for an optimal utility-privacy balance over the SNVs identified by MIG in this setting (on the order of $10^2$). In this scenario, the higher the value of $\theta$, the more SNVs the defender has to flip to guarantee privacy. A higher value of $\theta$ thus more clearly demonstrates the differences in utility across the methods, and therefore we present results for $\theta=0$ and $\theta=1000$.

\begin{figure}[ht]
	\centering
	\captionsetup[subfigure]{justification=centering}
	\begin{subfigure}[t]{0.45\columnwidth}
		\includegraphics[width=\textwidth]{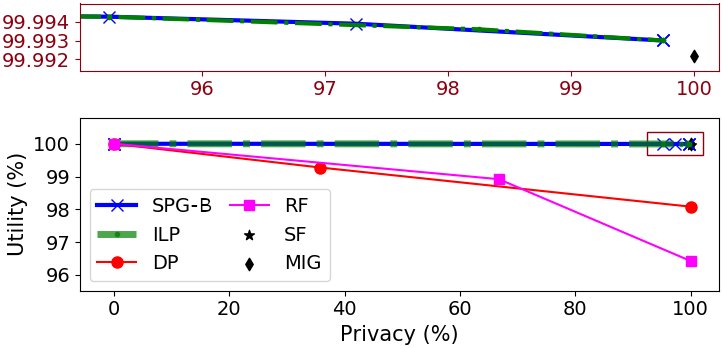}
		\caption{$\theta=0$}
		\label{fig:flip_theta_0}
	\end{subfigure}
	\begin{subfigure}[t]{0.45\columnwidth}
		\includegraphics[width=\textwidth]{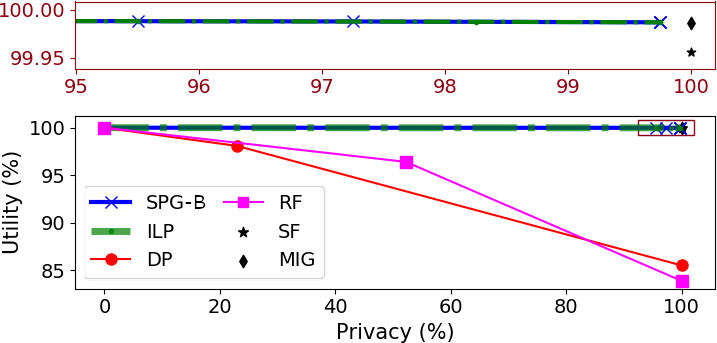}
		\caption{$\theta=1000$}
		\label{fig:flip_theta_1000}
	\end{subfigure}
	\caption{Fixed threshold attack model when all approaches only flip SNVs.}
	\label{fig:flip_theta}
    \vskip 10pt
\end{figure}

Fig.~\ref{fig:flip_theta} shows results in this setting when $\theta=0$ and $\theta=1000$. SPG-B again Pareto dominates DP, RF, and SF. 
SF guarantees privacy for all individuals while offering much lower utility. MIG guarantees privacy for all individuals with a slightly lower utility when $\theta=1000$, dropping further when $\theta=0$. In practice, the difference between the performance of MIG and SPG-B arises from flipping about 10 additional SNVs to guarantee privacy for only a single individual in the dataset. For a very large value of the weight parameter $w$, SPG-B produces the same solution as MIG in this setting.

\textsc{\textbf{Special Case: Only Masking SNVs}}
We now consider the alternative case where SNVs are only masked. This setting also serves to demonstrate the greater loss of utility that must be tolerated to achieve privacy for all using just masking. Note that the impact of masking a Beacon response is smaller than that of flipping it.

\begin{figure}[ht]
	\centering
	\captionsetup[subfigure]{justification=centering}
	\begin{subfigure}[t]{0.45\columnwidth}
		\includegraphics[width=\textwidth]{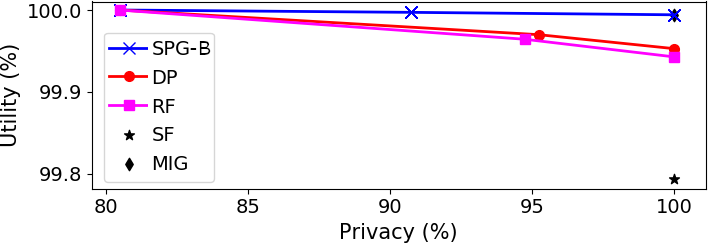}
		\caption{$\theta=-750$}
		\label{fig:mask_theta_n750}
	\end{subfigure}
	\begin{subfigure}[t]{0.45\columnwidth}
		\includegraphics[width=\textwidth]{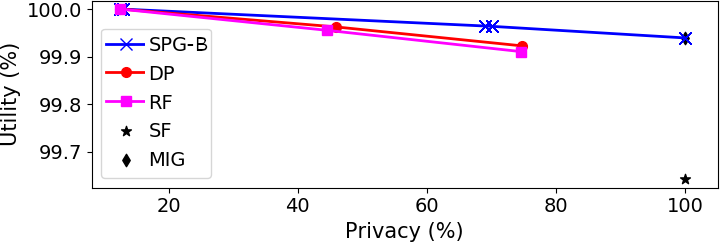}
		\caption{$\theta=-250$}
		\label{fig:mask_theta_n250}
	\end{subfigure}
	\caption{Fixed threshold attack model when all approaches only mask SNVs.}
	\label{fig:mask_theta}
 \vskip 10pt
\end{figure}

Fig.~\ref{fig:mask_theta} presents the results
%in this setting 
when the $\theta$ prediction threshold is set to -250 and -750. Yet again, we observe that SPG-B Pareto dominates all baselines, with MIG offering comparable utility when privacy of all individuals is necessarily guaranteed. Comparing the performance of SPG-B between Fig.~\ref{fig:mask_theta} and Fig.~\ref{fig:SPG-B_fix} for $\theta=-250$, it can be seen that choosing to flip a small number of SNVs and masking the remaining greatly improves utility.

\section{Additional Results - Adaptive Threshold}
\label{appx:ada_thresh}
Here, we present some additional results in the adaptive threshold setting for both SPG-B and SPG-R. First, we compare the performance of SPG-B to various baselines, when the threshold is set to the mean of $K=5$ lowest percentile of LRT scores, where baselines either flip or mask SNVs, while SPG-B combines both. From Fig.~\ref{fig:SPG-B_ada_5}, we can observe that SPG-B once again dominates all baselines. None of the baselines offer any privacy when they are restricted to masking SNVs in this case.

\begin{figure}[ht!]
    \centering
        \captionsetup[subfigure]{justification=centering}
        \begin{subfigure}[t]{0.45\columnwidth}
            \includegraphics[width=\textwidth]{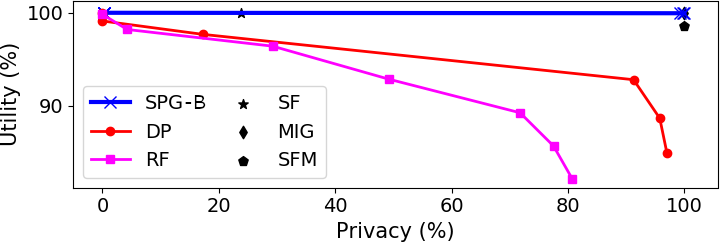}
            \caption{$K=5$, baselines only flip SNVs}
            \label{fig:mdf_flp_ada_5}
        \end{subfigure}
        \begin{subfigure}[t]{0.45\columnwidth}
            \includegraphics[width=\textwidth]{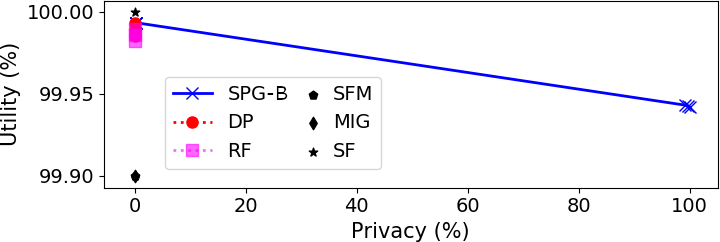}
            \caption{$K=5$, baselines only mask SNVs}
            \label{fig:mdf_msk_ada_5}
        \end{subfigure}
        \caption{Utility-privacy plots for the adaptive threshold attack model for AAF releases, compared to baselines.}
        \label{fig:SPG-B_ada_5}
        \vskip 10pt
\end{figure}

\textsc{\textbf{Special Case: Only Flipping SNVs}}\newline
Finally, we present results in the case where all methods, including SPG-B, only flip SNVs. Fig.\ref{fig:flip_ada} compares SPG-B to the baselines. It can be seen that  SPG-B offers a better privacy-utility balance than all methods, except for MIG when privacy of all individuals is to be guaranteed. In comparison to Figs.~\ref{fig:SPG-B_ada} and ~\ref{fig:SPG-B_ada_5}, it can be seen that SPG-B has a slightly lower utility. This further illustrates the value of applying a method that uses both flipping and masking.
			
			\begin{figure}[ht!]
				\centering
				\captionsetup[subfigure]{justification=centering}
				\begin{subfigure}[t]{0.45\columnwidth}
					\includegraphics[width=\textwidth]{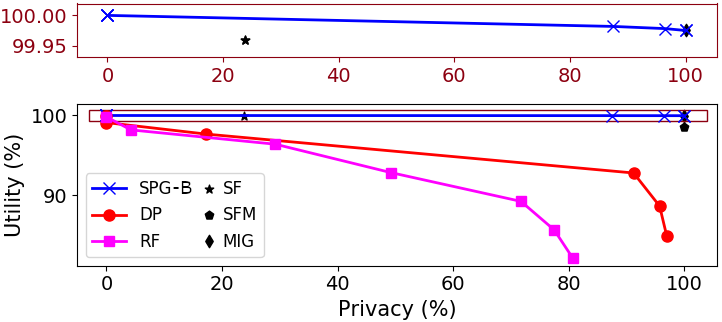}
					\caption{$K=5$}
					\label{fig:flip_ada_5}
				\end{subfigure}
				\begin{subfigure}[t]{0.45\columnwidth}
					\includegraphics[width=\textwidth]{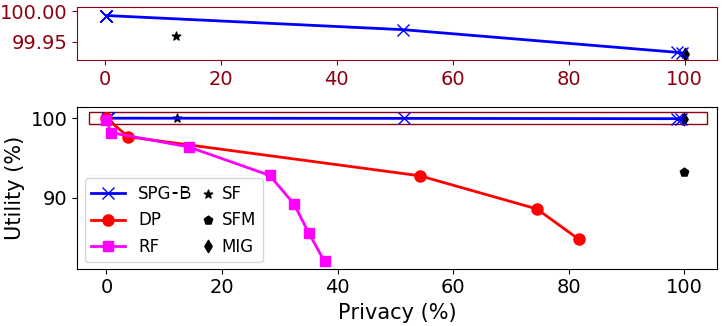}
					\caption{$K=10$}
					\label{fig:flip_ada_10}
				\end{subfigure}
				\caption{Adaptive threshold attack model, where all approaches only flip SNVs.}
				\label{fig:flip_ada}
                \vskip 10pt
			\end{figure}

\section{Additional Results - Bounded Risk}
\label{appx:bounded}
Here, we compare the performance of SPG-R (parallel) to the various baselines, under the assumption of bounded risk, where the sensitivity of the mean query is calculated in the average case instead of the worst case, and correspondingly the scale of the Laplacian depends on the average number of bits by which a genome in the dataset differs from those not in the dataset. On our data, the average sensitivity is $148515$, which is an order of magnitude smaller than the number of SNVs.

This has no qualitative impact on the best solutions found by our approach, except that a correspondingly smaller value of DP parameter $\epsilon$ is now used to generate the same amount of noise. Because our candidate $\epsilon$ values were well spread out ($\epsilon\in\{10K, 50K, 100K, 500K, 1M, 5M, 10M\}$), our approach works well without any modifications. 

\textsc{\textbf{Fixed Threshold Attacks}}\newline
Fig.~\ref{fig:fixed_bounded} compares SPG-R to the various baselines under the fixed threshold attack model, when $\theta=0$ for $\alpha=0.5$ and $\alpha=0.9$. As is the case with unbounded risk, SPG-R dominates all baselines when $\alpha=0.5$, but as masking gets relatively cheaper compared to adding noise, SPG-R is dominated by a pure-masking strategy, with the difference being more pronounced in the bounded risk scenario.

\begin{figure}[ht!]
    \centering
    \captionsetup[subfigure]{justification=centering}
    \begin{subfigure}[t]{0.45\columnwidth}
        \includegraphics[width=\textwidth]{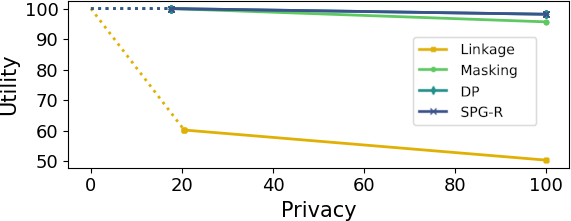}
        \caption{$\alpha=0.5$}
        \label{fig:fixed_bounded_0.5}
    \end{subfigure}
    \begin{subfigure}[t]{0.45\columnwidth}
        \includegraphics[width=\textwidth]{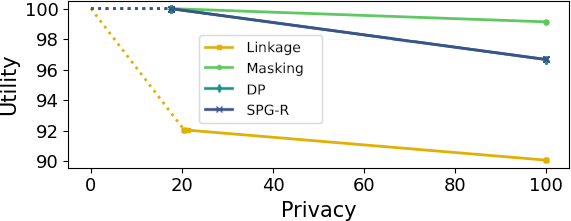}
        \caption{$\alpha=0.9$}
        \label{fig:fixed_bounded_0.9}
    \end{subfigure}
    \caption{Performance of SPG-R compared to baselines in the fixed threshold setting (bounded risk) $\theta=0$.}
    \label{fig:fixed_bounded}
    \vskip 10pt
\end{figure}
\smallskip

\textsc{\textbf{Adaptive Threshold Attacks}}\newline
The performance with bounded risk in the adaptive threshold setting yet again qualitatively mirrors the results in the unbounded risk setting, with SPG-R dominating all baselines, as we can observe from Fig.~\ref{fig:ada_bounded}. 

\begin{figure}[ht!]
    \centering
    \captionsetup[subfigure]{justification=centering}
    \begin{subfigure}[t]{0.45\columnwidth}
        \includegraphics[width=\textwidth]{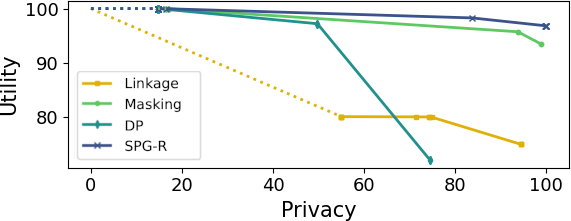}
        \caption{$K=5, \alpha=0.75$}
        \label{fig:ada_bounded_5_0.75}
    \end{subfigure}
    \begin{subfigure}[t]{0.45\columnwidth}
        \includegraphics[width=\textwidth]{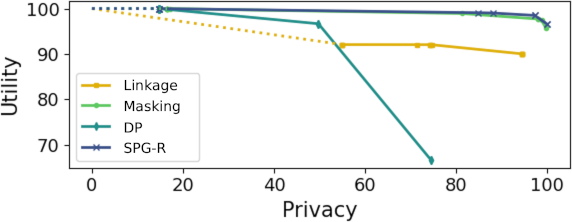}
        \caption{$K=5, \alpha=0.9$}
        \label{fig:ada_bounded_5_0.9}
    \end{subfigure}
    
    \begin{subfigure}[t]{0.45\columnwidth}
        \includegraphics[width=\textwidth]{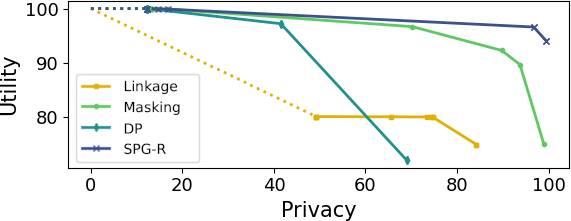}
        \caption{$K=10, \alpha=0.75$}
        \label{fig:ada_bounded_10_0.75}
    \end{subfigure}
    \begin{subfigure}[t]{0.45\columnwidth}
        \includegraphics[width=\textwidth]{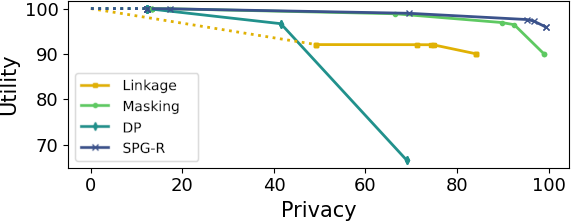}
        \caption{$K=10, \alpha=0.9$}
        \label{fig:ada_bounded_10_0.9}
    \end{subfigure}
    \caption{Performance of SPG-R compared to baselines in the adaptive threshold setting (bounded risk) for $K=5$ and $K=10$.}
    \label{fig:ada_bounded}
    \vskip 10pt
\end{figure}

\begin{figure}[ht!]
    \centering
    \includegraphics[width=0.45\columnwidth]{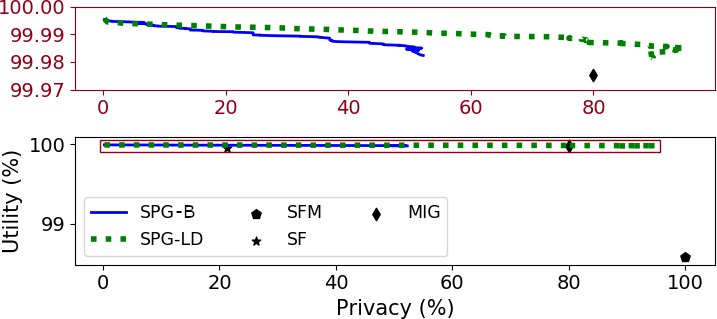}
    \caption{Adaptive threshold attack model, where the attacker leverages correlation data, and baselines only flip SNVs. The Greedy approach permitted to flip or mask all SNVs. $K=5$}
    \label{fig:LD_ada_flip_noDNF}
    \vskip 10pt
\end{figure}

\section{Additional Results - LD}
\label{appx:LD}

In the adaptive threshold case, recall that flipping or masking is restricted to SNVs where $\Delta^{F(K)}_{ij}\ge 0$ and $\Delta^{M(K)}_{ij}\ge 0$ respectively for all individuals in the Beacon. When a SNV $j$ is flipped or masked by \emph{SPG-LD}, there may be SNVs in $N_{LD}(j)$ for which these inequalities may not hold true, and are therefore not flipped or masked. These SNVs however may still be used to infer flipped or masked Beacon responses by measuring correlations.
    
The correlation attack had no impact on SPG-BB when $K=10$. None of the SNVs picked by SPG-B in this setting to either flip or mask had correlations with other SNVs within a sliding window of $500$ SNVs ($250$ on each side). Therefore, we ran a second set of experiments that neglects these inequality constraints, essentially allowing Beacon responses for all SNVs to be flipped or masked, with the consequence that privacy achieved in early iterations of the greedy algorithms may be reduced by later flips or masks. Fig.~\ref{fig:LD_ada_flip_noDNF} presents the results. While the performance of SF, SFM and MIG are unchanged compared to the previous setting in Fig.~\ref{fig:LD_ada}, notice that for both SPG-B and SPG-LD, the privacy starts decreasing after a point, as more and more SNVs are flipped or masked. The maximum privacy achieved by the modified greedy algorithm accounting for correlations, SPG-LD, is around $90\%$. This is a marginal increase compared to the previous setting.

\edited{
    \section{Impact of Increasing Population Sizes}
    Here, we present additional results comparing the relative performance of DP to SPG-B and SPG-R for Beacons and summary statistics as we increase the number of individuals in the dataset ($n$). With the available hardware, we were able to experiment with summary statistics for datasets consisting of up to 600 individuals. Fig.~\ref{fig:DPvsSPGB} shows that neither approach is significantly affected by population size in the case of Beacons. Fig.~\ref{fig:DPvsSPGR} suggests that, while neither approach is affected in the fixed-threshold setting for summary statistics, the relative performance of DP deteriorates in terms of utility as the number of individuals increases in the adaptive threshold case. The slight difference in performance trends between these figures and the  results presented in the main paper is due to randomly splitting the universe of all individuals ($D\cup \bar{D}$) into two sets of varying sizes here, as compared to working with a particular 400-400 split in previous settings.}

    \begin{figure}
        \centering
        \begin{subfigure}[t]{0.45\columnwidth}
            \includegraphics[width=\textwidth]{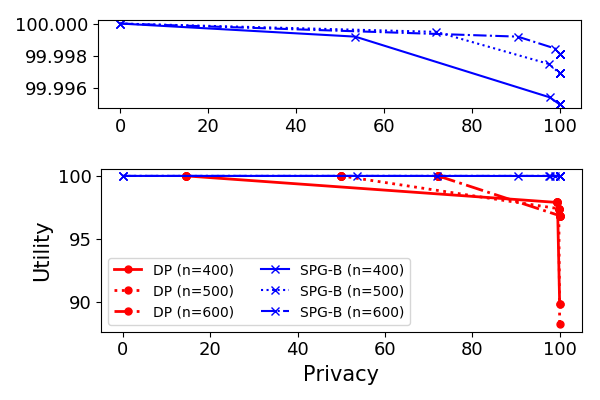}
            \caption{$\theta=-250$, DP flips SNVs}
            \label{fig:SPGvDP_fix_beacon}
        \end{subfigure}
        \begin{subfigure}[t]{0.45\columnwidth}
            \includegraphics[width=\textwidth]{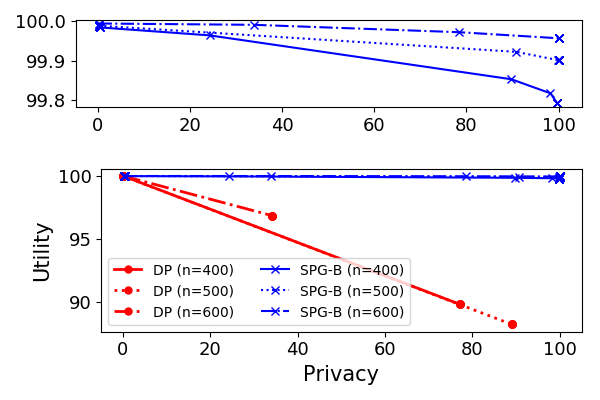}
            \caption{$K=10$, DP flips SNVs}
            \label{fig:SPGvDP_ada_beacon}
        \end{subfigure}
        \caption{Relative performance of DP and SPG-B with increasing dataset size ($n$). Zoomed-in portions shown on top.}
        \label{fig:DPvsSPGB}
        \vskip 10pt
    \end{figure}
    
    \begin{figure*}[ht!]
    \centering
        \begin{subfigure}[t]{0.45\columnwidth}
            \includegraphics[width=\textwidth]{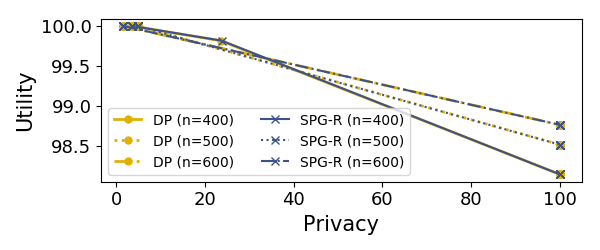}
            \caption{$\theta=0, \alpha=0.5$}
            \label{fig:SPGvDP_fix_alpha_0.5}
        \end{subfigure}
        \begin{subfigure}[t]{0.45\columnwidth}
            \includegraphics[width=\textwidth]{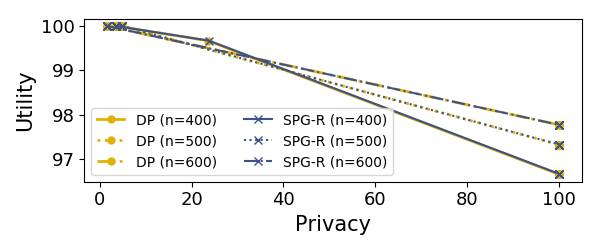}
            \caption{$\theta=0, \alpha=0.9$}
            \label{fig:SPGvDP_fix_alpha_0.9}
        \end{subfigure}
        
        \begin{subfigure}[t]{0.45\columnwidth}
            \includegraphics[width=\textwidth]{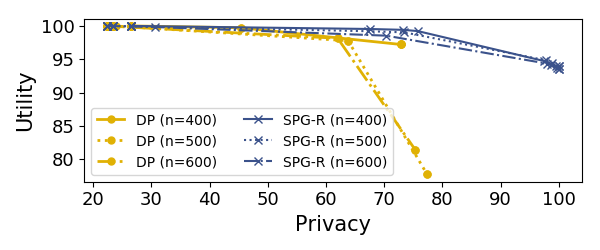}
            \caption{$K=10, \alpha=0.75$}
            \label{fig:SPGvDP_ada_alpha_0.75}
        \end{subfigure}
        \begin{subfigure}[t]{0.45\columnwidth}
            \includegraphics[width=\textwidth]{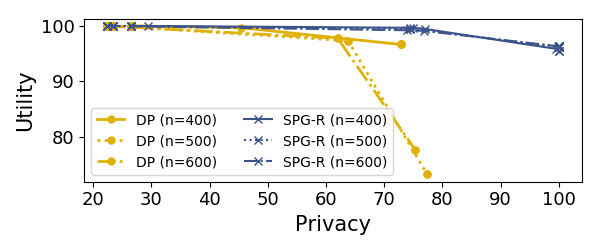}
            \caption{$K=10, \alpha=0.9$}
            \label{fig:SPGvDP_ada_alpha_0.9}
        \end{subfigure}
    \caption{Relative performance of DP and SPG-R with increasing dataset size ($n$)}
    \label{fig:DPvsSPGR}
    \vskip 10pt
    \end{figure*}

\section{Empirical Runtime - ILP vs SPG-B}
Finally, we present empirical results comparing the running times of the ILP and SPG-B (refer Section~\ref{appx:fix_thresh} for details). Because the ILP was only feasible over a small set of SNVs chosen by MIG, we restrict SPG-B to the same set of SNVs to ensure a fair comparison. For both approaches, running time was only measured for the steps involved in solving the optimization, ignoring all initialization and pre-computation steps. This analysis was performed on a 2018 MacBook Pro with an Intel i7 processor and 16 GB of RAM.

Fig.~\ref{fig:runtimes} depicts that the running time for the ILP is between 2 to 6 orders of magnitudes larger, depending on the weight parameter. For small weights, the optimal solution is to do nothing, and for large enough weights, the optimal solution is to guarantee privacy for everyone - these are solutions that a branch-and-bound approach is expected to arrive quickly at - which is what IBM CPLEX uses in order to compute optimal solutions. For weight parameters between these two extremes, we expect that the branching would need to continue to a greater depth, leading to the spike in running time at $w=0.2$. On the other hand, SPG-B continues flipping or masking SNVs until all individuals are covered, regardless of $w$, the weight is only used to update the current best solution which is eventually returned. Therefore the running time is constant across all weights.

\begin{figure*}[b!]
    \centering
    \includegraphics[width=0.8\textwidth]{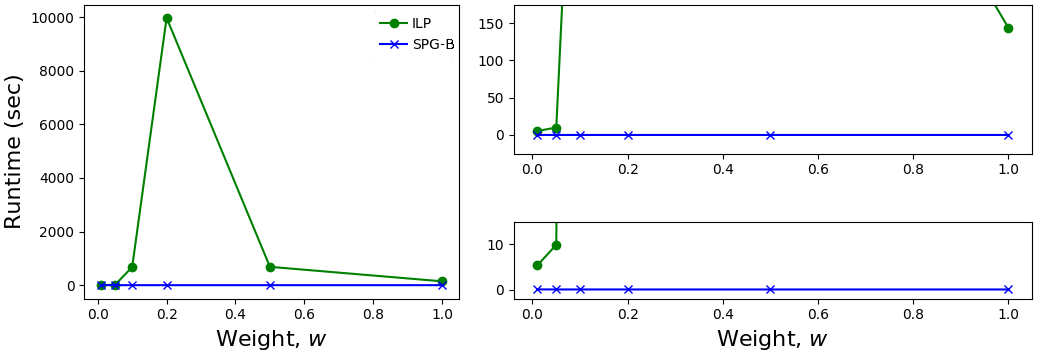}
    \caption{Empirical Comparison of running times for ILP and SPG-B (using only SNPs identified by MIG)}
    \label{fig:runtimes}
\end{figure*}
\end{document}